%% file: main.tex
\documentclass[journal,twoside,web,final]{ieeecolor}

\newif\ifmainfile
\mainfiletrue

\newif\ifcompile
\compiletrue

\newif\ifBio
\Biofalse
\usepackage{pbalance}

\input{header.tex}

\input{linestyles.tex}

\begin{document}


\title{Policy Optimization with Differentiable MPC: \\ Convergence Analysis under Uncertainty}

\author{
    Riccardo Zuliani\textsuperscript{1}, %
    Efe C. Balta\textsuperscript{1,2}, %
    and John Lygeros\textsuperscript{1} %
    \thanks{
        This work was supported as a part of NCCR Automation, 
        a National Centre of Competence in Research, 
        funded by the Swiss National Science Foundation (grant number 51NF40\_225155). 
        \textsuperscript{1}Automatic Control Laboratory (IfA), 
        ETH Z\"urich, 8092 Z\"urich, Switzerland 
        \texttt{\small$\{$rzuliani,lygeros$\}$@ethz.ch}. 
        \textsuperscript{2}Control and Automation Group, inspire AG, 
        8005 Z\"urich, Switzerland. 
        \texttt{\small efe.balta@inspire.ch}.
    }   
}

\maketitle

\begin{abstract}
Model-based policy optimization is a well-established framework for designing reliable and high-performance %
controllers across a wide range of control applications. 
Recently, this approach has been extended to model predictive control policies, 
where explicit dynamical models are embedded within the control law. 
However, the performance of the resulting controllers, 
and the convergence of the associated optimization algorithms, critically depends on the accuracy of the models.
In this paper, we demonstrate that combining gradient-based policy optimization with recursive system %
identification ensures convergence to an optimal controller and showcase our finding in several examples.
\end{abstract}



\section{Introduction}
\label{section:introduction}
\input{sources/introduction.tex}

\section{Preliminaries}
\label{section:prelim}
\input{sources/preliminaries.tex}

\section{Problem formulation}
\label{section:pform}
\input{sources/problem_formulation.tex}

\section{Proposed method}
\label{section:pref}
\input{sources/proposed_method.tex}

\subsection{Learning the true model}
\label{section:learning}
\input{sources/learning_model.tex}

\subsection{Computing sensitivities using backpropagation}
\label{section:backprop}
\input{sources/backprop.tex}

\subsection{The Complete Algorithm}
\label{section:algorithm}
\input{sources/algorithm.tex}

\section{Convergence analysis}
\label{section:convergence}
\input{sources/convergence.tex}

\section{Special Cases and Extensions}
In this section, we introduce two variants of the proposed algorithm: 
one in which the MPC model $\tilde{\theta}^k$ is always set equal to the current best estimate %
$\theta^k$ (see \cref{section:equivalence}), 
and another in which the persistency of excitation assumption is removed (see \cref{section:imperfect_learning}).

\subsection{Certainty equivalence}
\label{section:equivalence}
\input{sources/certainty_equivalence.tex}

\subsection{Convergence with imperfect model learning}
\label{section:imperfect_learning}
\input{sources/imperfect_learning.tex}

\section{Simulation Results}
\label{section:simulation_results}
\input{sources/simulation_results.tex}

\section{Conclusion}
\label{section:conclusion}
\input{sources/conclusion.tex}


\bibliographystyle{IEEEtran}
\bibliography{bibliography.bib}

\appendix
\crefalias{subsubsection}{appendix}%
\crefalias{subsection}{appendix}%
\crefalias{section}{appendix}

\subsection{Proof of Theorem \ref{thm:main}}
\label{section:convergence_proof}
\input{sources/convergence_proof_new.tex}
\subsection{Path-differentiability of optimization problems}
\label{section:path_diff_opt_prob}
\input{sources/path_diff.tex}

\ifBio
\input{sources/biography.tex}
\fi

\end{document}

%% file: linestyles.tex
\definecolor{red1}{HTML}{911B17} 
\definecolor{red2}{HTML}{A9221E} 
\definecolor{red3}{HTML}{DD322C} 
\definecolor{red4}{HTML}{FCE9E8} 

\definecolor{blue1}{HTML}{004494} 
\definecolor{blue2}{HTML}{005DC7} 
\definecolor{blue3}{HTML}{2180ED} 
\definecolor{blue4}{HTML}{5A9CE8} 
\definecolor{blue5}{HTML}{D1E4FA} 
\definecolor{blue6}{HTML}{E8F2FC} 

\definecolor{green1}{HTML}{275d4c} 
\definecolor{green2}{HTML}{388069} 
\definecolor{green3}{HTML}{48A386} 
\definecolor{green4}{HTML}{B7E6D7} 
\definecolor{green5}{HTML}{CAECE2} 

\definecolor{yellow1}{HTML}{CFA621} 
\definecolor{yellow2}{HTML}{DEBB44} 
\definecolor{yellow3}{HTML}{E5CA72} 
\definecolor{yellow4}{HTML}{F3E7BF} 
\definecolor{yellow5}{HTML}{FBF5E6} 

\definecolor{orange1}{HTML}{CC7224} 
\definecolor{orange2}{HTML}{DC8D47} 
\definecolor{orange3}{HTML}{E2A774} 
\definecolor{orange4}{HTML}{F2D7C0} 
\definecolor{orange5}{HTML}{FAEFE6} 

\definecolor{gray1}{HTML}{0D0D0D} 
\definecolor{gray2}{HTML}{1A1A1A} 
\definecolor{gray3}{HTML}{333333} 
\definecolor{gray4}{HTML}{595959} 
\definecolor{gray5}{HTML}{666666} 
\definecolor{gray6}{HTML}{B3B3B3} 
\definecolor{gray7}{HTML}{D9D9D9} 
\definecolor{gray8}{HTML}{F2F2F2} 
\definecolor{gray9}{HTML}{FFFFFF} 

\colorlet{arrowColor}{gray6!50!gray5}
\colorlet{arrowColorText}{gray6!50!gray5}
\colorlet{deployColor}{blue1}
\colorlet{sensitivityColor}{orange1}
\colorlet{sysIdColor}{orange1}
\colorlet{updateColor}{orange1}
\colorlet{newArrowColor}{red1}

\colorlet{path_border_color}{gray3}
\colorlet{path_fill_color}{gray7!50!gray8}
\colorlet{landmark_color}{gray1}

\colorlet{best_color_table}{yellow2}
\colorlet{draw_color_table}{gray1}
\colorlet{trained_color_table}{blue1}
\colorlet{table_dim_gray}{gray1}
\colorlet{table_dark_gray}{gray6}

\pgfplotsset{
    axis line style = {color=gray1,semithick},
    tick style = {color=gray1,semithick},
    tick label style = {color=gray1},
    xlabel style = {color=gray1},
    ylabel style = {color=gray1},
    legend style = {draw=gray1,style={text=gray1},semithick},
    grid style = {gray7}
}

\pgfplotsset{solid/.style={
    ultra thick
}}
\pgfplotsset{dashed/.style={
    very thick, dash pattern=on 5pt off 2pt
}}
\pgfplotsset{dashdot/.style={
    very thick, dash pattern=on 5.25pt off 1.75pt on 1.75pt off 1.75pt
}}
\pgfplotsset{mark_square/.style={
    mark=square*, mark size=1pt, mark options={solid}
}}

\pgfplotsset{energy_informed/.style={dashdot,blue1}}
\pgfplotsset{nominal/.style={solid,gray6}}

\pgfplotsset{best/.style={dashdot,orange1}}
\pgfplotsset{trained/.style={solid,blue1}}
\pgfplotsset{untrained/.style={solid,gray6}}

\pgfplotsset{untrained_x/.style={dashdot, blue1, mark_square, mark repeat=10}}
\pgfplotsset{best_x/.style={dashed, blue1!45, forget plot}}
\pgfplotsset{trained_x/.style={solid, blue1!60, forget plot}}

\pgfplotsset{untrained_y/.style={dashdot, red1, mark_square, mark repeat=10}}
\pgfplotsset{best_y/.style={dashed, red1!45, forget plot}}
\pgfplotsset{trained_y/.style={solid, red1!60, forget plot}}

\pgfplotsset{untrained_z/.style={dashdot, orange1, mark_square, mark repeat=10}}
\pgfplotsset{best_z/.style={dashed, orange1!45, forget plot}}
\pgfplotsset{trained_z/.style={solid, orange1!60, forget plot}}

\pgfplotsset{iteration_best/.style={dashdot,orange1}}
\pgfplotsset{iteration_exact/.style={solid,gray6}}
\pgfplotsset{iteration_inexact/.style={dashed,blue1}}
\pgfplotsset{estimation_error/.style={solid,red1}}

\pgfplotsset{untrained_path_color/.style={untrained,gray5!50!gray6}}
\pgfplotsset{trained_path_color/.style={trained}}
\pgfplotsset{best_path_color/.style={best}}

%% file: sources/introduction.tex
\IEEEPARstart{M}{odel} predictive control (MPC) is a well-established method that uses a model of the system dynamics to compute control actions in real time. At each time step, the current system state is measured, an optimal control problem is solved over a finite prediction horizon, and the first optimal input is applied to the system. The procedure is then repeated at the next time step, thereby enabling feedback.

Designing the cost function and the constraints in the MPC has been thoroughly studied in the literature. 
While the MPC cost can in principle match (or closely approximate) the prescribed higher-level objective (for example, in an economic MPC scheme), it is often chosen to be a quadratic function of the decision variables to enhance numerical robustness and make the optimization problem easier to solve.
The constraints are generally enforced to guarantee system safety both within the MPC horizon and beyond. 
For systems with linear dynamics and quadratic cost functions, it is conventional to introduce a terminal cost obtained by solving the discrete-time Riccati equation to approximate the infinite horizon cost. 
This technique was first introduced in \cite{sznaier1987suboptimal}, and later extended to nonlinear dynamics in \cite{chen1998quasi} by linearizing around the target equilibrium point. 
While these techniques ensure favorable properties in closed-loop, such as stability and recursive feasibility \cite{scokaert2002constrained}, they generally introduce suboptimality.
Several techniques have been proposed to choose the cost function of the MPC in different settings, such as tracking periodic references or utilizing artificial setpoints. We refer the reader to \cite{kohler2024analysis} for an overview.

An alternative and increasingly popular approach involves formulating the design of the MPC as a closed-loop optimization problem. This corresponds to a policy optimization problem, where the policy is encoded implicitly through the MPC controller.
Several works propose Bayesian Optimization (BO) to solve this policy optimization problem \cite{edwards2021automatic,sorourifar2021data,puigjaner2025performance}. Others leverage differentiable optimization \cite{amos2017optnet,amos2018differentiable}, which enables the computation of gradients of the solution map to an optimization problem with respect to its parameters. In this work, we adopt the latter approach.
Our work builds on the recently proposed BP-MPC framework \cite{zuliani2023bp,zuliani2024closed} that enables optimization of MPC policies with convergence guarantees.

The idea of automatically tuning the hyper-parameters of MPC with differentiable optimization dates back to the seminal OptNet paper of Amos \& Kolter \cite{amos2017optnet}, which showed how a quadratic program can act as a differentiable layer inside a neural network. Shortly after, the same authors generalized the approach to the receding-horizon setting in \cite{amos2018differentiable}. A complementary line of work appeared in \cite{agrawal2020learning}, where the authors formalized the concept of implicit differentiation of convex cone programs. More recent developments include \cite{oshin2023differentiable}, which brings end-to-end gradient information to tube-based robust controllers; \cite{drgovna2022differentiable}, where a physics-informed neural state-space model and its MPC policy are trained jointly; \cite{tao2024difftune}, using the same approach as \cite{amos2018differentiable}; and the very recent \cite{frey2025differentiable}, which shows how to compute sensitivities of nonlinear optimization problems solved using sequential quadratic programming.

Most of these papers formulate the problem as either supervised learning or reinforcement learning, and typically lack the recursive feasibility and Lyapunov stability guarantees. This limitation is addressed by Gros and Zanon, beginning with the data-driven economic MPC framework \cite{gros2019data}, where the authors combine policy gradient methods with nonlinear MPC, using the value function of the MPC as an approximator of the optimal cost-to-go in a reinforcement learning problem. Subsequent works combined this framework with tube-based MPC to ensure safety \cite{gros2020safe,gros2022learning} and stability \cite{zanon2019practical}. These works, however, do not address the question of convergence.

It is worth noting that differentiating the KKT conditions via the Implicit Function Theorem (IFT) is much older than the deep-learning literature, dating back to the 1960s \cite{Fiacco1990nonlinear,jittorntrum1978sequential}. Within the optimal-control community, systematic treatments of NLP sensitivities can be found in \cite{pirnay2012optimal}, \cite{andersson2018sensitivity}, and the open-source interior-point implementation in \cite{frey2025differentiable}.

Despite this growing body of work, the literature still lacks a rigorous treatment of policy optimization with MPC when the system dynamics are only partially known. To the best of the author's knowledge, no convergence results exist for this setting.
In this work, we extend the convergence guarantees previously established in \cite{zuliani2023bp} to the case where the system dynamics are uncertain and affected by stochastic noise. Our contributions are twofold:
i) we provide convergence guarantees under asymptotically exact system identification, and
ii) we establish convergence to a suboptimal solution when the model is not learned perfectly.
Providing convergence guarantees in this context remains an open problem largely unaddressed in prior literature, with the sole exception of our previous work \cite{zuliani2023bp}, which considered convex MPC architectures under full model knowledge. Here, we explicitly account for stochastic noise and model uncertainty.
Finally, this work does not focus on safety aspects, which were recently addressed in \cite{zuliani2024closed}.

The remainder of this paper is organized as follows. In \cref{section:prelim}, we briefly review the notions of path-differentiability and definability. The problem formulation is presented in \cref{section:pform}, followed by a description of the proposed method in \cref{section:pref} and an analysis of its convergence properties in \cref{section:convergence}. In \cref{section:equivalence} and \cref{section:imperfect_learning}, we introduce two variants of our approach, the first using the certainty-equivalence principle and the second addressing the case of imperfect model knowledge, respectively. Finally, \cref{section:simulation_results} reports our simulation results.

\subsubsection*{Notation}
\label{section:notation}
We use $\Z$ to denote the set of integers, and set $\Z_{[a,b]}= \Z \cap \{ x: a \leq x \leq b \}$. We denote the sets of real and natural numbers with $\R$ and $\N$, respectively. The standard Euclidean norm is denoted with $\|\cdot\|$, and given a symmetric positive definite matrix $A$ we define $\|x\|_A = \sqrt{x^\top A x}$. We denote by $\mathbb{E}_w[\cdot]$ the expectation with respect to the random variable $w$. The standard Euclidean distance between a vector $x\in \R^n$ and a set $\mathcal{X} \subset \R^n$ is denoted with $\operatorname*{dist}(x,\mathcal{X})$. We denote with $(a_n)_{n\in\N}$ the sequence $a_0,a_1,\dots$.

%% file: sources/preliminaries.tex
The concept of path differentiability \cite{bolte2021conservative} extends the notion of differentiability to almost everywhere differentiable locally Lipschitz functions. Given a locally Lipschitz function $f:\R^n\to\R^m$, and a compact-valued outer semicontinuous set-valued function $\mathcal{J}_f:\R^n \rightrightarrows \R^{m \times n}$, we say that $f$ admits $\mathcal{J}_f$ as a conservative Jacobian if, for all absolutely continuous curves $\theta:[0,1]\to\R^n$ and almost every $t\in[0,1]$, one has
\begin{align}
\frac{d}{dt}f(\theta(t)) = V\dot{\theta}(t), ~~ \forall V\in \mathcal{J}_f(\theta(t)). \label{eq:path_diff_definition}
\end{align}
In this case, we say that $f$ is path differentiable. 
Given a function $f:\R^n \times \R^p \to \R^m$, we define $\J_{f,x}(x,y):=\{ J_x\in\R^{m \times n} : \exists J_y\in\R^{m \times p}, [J_x~J_y]\in\J_f(x,y) \}$, and similarly for $\J_{f,y}$. 
Conservative Jacobians are almost everywhere equal to standard Jacobians, and they obey several useful properties, like the chain rule of differentiation and a nonsmooth implicit function theorem. 
We refer the reader to \cite{bolte2021conservative} for an overview.

The class of path differentiable functions is quite broad and it comprises all functions that are definable in an o-minimal structure \cite[Proposition 2]{bolte2021conservative}. An \emph{o-minimal structure} expanding the real field $\R$ is a collection of sets $\mathcal{S}=(\mathcal{S}^n)_{n\in\N}$, with each $\mathcal{S}^n \subset \R^n$ such that
\begin{enumerate}
    \item all algebraic subsets of $\R^n$ are contained in $\mathcal{S}^n$;
    \item $\mathcal{S}^n$ is a Boolean subalgebra of $\R^n$;
    \item if $A\in \mathcal{S}^m$ and $B\in \mathcal{S}^m$, then $A \times B \in \mathcal{S}^{n+m}$;
    \item the projection onto the first $n$ coordinates of any $A\in \mathcal{S}^{n+1}$ belongs to $\mathcal{S}^n$;
    \item the elements of $\mathcal{S}^1$ are precisely the finite unions of points and intervals.
\end{enumerate}
The elements of $\mathcal{S}^n$ are called \emph{definable subsets} of $\R^n$, and a function is called \emph{definable} (in an o-minimal structure) if its graph is a definable set. 
The majority of functions encountered in controls and optimization are definable. 
Definability is preserved by addition, multiplication, differentiation, integration, and composition.

%% file: sources/problem_formulation.tex
We consider a discrete-time system controlled by an MPC
\begin{align}
\begin{split}
x_{t+1}&=f(x_t,u_t,\theta) + w_t,\\
y_t&=\operatorname{MPC}(x_t,y_{t-1},p),\\
u_t&=\pi(x_t,y_t,p),
\end{split}\label{eq:system}
\end{align}
for $t\in\Z_{[0,T-1]}$. In \cref{eq:system}, $p\in \mathcal{P}$ is a tunable design parameter chosen from a design set $\mathcal{P}\subset \R^{n_p}$, 
$y_t \in \R^{n_y}$ is the optimizer of the MPC problem at time $t$, 
and $x_t\in\R^{n_x}$ and $u_t\in\R^{n_u}$ are the system state and input, respectively, with a given $x_0$. 
The control law $\pi:\R^{n_x} \times \R^{n_y} \times \mathcal{P}\to \R^{n_u}$ generates the input based on the current state and MPC output. 
The term $w_t\in\R^{n_x}$ is an unknown random disturbance (for detailed assumptions see \cref{section:learning}). 
The system dynamics $f$ depend on an unknown but constant parameter $\theta\in\R^{n_\theta}$, and are modeled as
\begin{align}
f(x,u,\theta) = \phi(x,u)^{\top}\theta + \varphi(x,u), \label{eq:true_dynamics}
\end{align}
where $\phi:\R^{n_x} \times \R^{n_u} \to \R^{n_x \times n_\theta}$ and $\varphi:\R^{n_x}\times \R^{n_u}\to \R^{n_x}$ are known, possibly nonlinear maps. 
The state and input must fulfill the following constraints throughout the horizon
\begin{align}
\begin{split}
u_t & \in \mathcal{U} := \{ u \in \R^{n_u} : H_u u \leq h_u \}, \\
x_t & \in \mathcal{X} := \{ x \in \R^{n_x} : H_x x \leq h_x \}.
\end{split} \label{eq:constraints}
\end{align}
where $\mathcal{U} \subseteq \R^{n_u}$ and $\mathcal{X} \subseteq \R^{n_x}$ are known constraint sets. 
Our goal is to identify a parameter $p$ that minimizes an upper-level cost function, while learning the unknown parameter $\theta$ using online system identification techniques.

Following our previous work \cite{zuliani2023bp}, we focus on linear MPC architectures formulated as quadratic programs of the form
\begin{align}
\operatorname*{minimize}_{x_{\cdot|t},u_{\cdot|t},\epsilon_{\cdot|t}} & ~~  P_\epsilon(\epsilon_{\cdot|t},p) + \ell_N(x_{N|t},p) + \sum_{j=0}^{N-1} \ell_j(x_{j|t},u_{j|t},p) \notag \\
\text{subject to} & ~~ x_{j+1|t}=A_j(y_{t-1},p) x_{j|t} + B_j(y_{t-1},p) u_{j|t} \notag \\ & ~~ + c_j(y_{t-1},p),~ \forall j\in\Z_{[0,N-1]}, \notag \\ 
& ~~ H_x x_{j|t} \leq h_x + \epsilon_{j|t} ,~\epsilon_{j|t} \geq 0,~ \forall j\in\Z_{[0,N]}, \notag \\
& ~~ H_u u_{j|t} \leq h_u,~ \forall j\in\Z_{[0,N-1]}, \notag \\
& ~~ x_{0|t} = x_t, \label{eq:MPC}
\end{align}
where $x_{j|t}$, $u_{j|t}$, and $\epsilon_{j|t}$ denote, respectively, the $j$-step-ahead predictions of the state, input, and slack variables computed at time $t$, $y_{t-1}=(x_{\cdot|t-1},u_{\cdot|t-1},\epsilon_{\cdot|t-1})$, and $\ell_j(\cdot,p)$ are strongly convex quadratic functions for every $p\in\mathcal{P}$ and $j$. 
The problem is guaranteed to be feasible for any parameter configuration thanks to the slack variable $\epsilon$ relaxing the constraints. 
To discourage constraint violations, the cost function includes a penalty term
\begin{align}
P_\epsilon(\epsilon,p) = c_1(p) \mathds{1}^\top \epsilon + c_2(p) \epsilon^\top \epsilon, \label{eq:problem_formulation:mpc_slack_penalty}
\end{align}
where $c_1(p),c_2(p)>0$ may optionally depend on $p$ \cite{kerrigan2000soft}. 
\rz{The method presented here can be adapted to the case where $\mathcal{X}$ and $\mathcal{U}$ are nonconvex by replacing the inequality constraints in \cref{eq:MPC} with $H_x(y_{t-1},p) x_{j|t} \leq h_x(y_{t-1},p)$ and $H_u(y_{t-1},p) u_{j|t} \leq h_u(y_{t-1},p)$, 
where $H_x(\cdot)$, $H_u(\cdot)$, $h_x(\cdot)$, and $h_u(\cdot)$ are trainable functions. 
By allowing the constraints to depend on $p$, the MPC could, through policy optimization, learn to satisfy the nonconvex constraints while maintaining a convex formulation. 
For simplicity, we assume polytopic constraints and leave the general case for future work.}

The system dynamics in \cref{eq:MPC} are enforced through the affine equality constraint
\begin{align}
x_{j+1|t}=A_j(y_{t-1},p) x_{j|t} + B_j(y_{t-1},p) u_{j|t} + c_j(y_{t-1},p), \label{eq:MPC_prediction_model}
\end{align}
where $A_j$, $B_j$, and $c_j$ are obtained for all $j\in\Z_{[0,N-1]}$ by linearizing $f$ along the trajectory defined by the previous MPC solution $y_{t-1}$ using an estimate $\tilde{\theta}$ of $\theta$
\begin{align}
A_j(y_{t-1},p) & = \nabla_x f(x_{j+1|t-1},u_{j+1|t-1},\tilde{\theta}), \notag \\
B_j(y_{t-1},p) & = \nabla_u f(x_{j+1|t-1},u_{j+1|t-1},\tilde{\theta}), \label{eq:linearization_along_trajectory} \\
c_j(y_{t-1},p) & = f(x_{j+1|t-1},u_{j+1|t-1},\tilde{\theta}) \notag \\ & - A_j(y_{t-1},p)x_{j+1|t-1} - B_j(y_{t-1},p)u_{j+1|t-1}. \notag
\end{align}
Optionally, $x_{1|t-1}$ in \cref{eq:linearization_along_trajectory} can be replaced with the current state $x_{t}$.
We assume, for now, that the design parameter $p=(\tilde{p},\tilde{\theta})$ includes both the nominal model $\tilde{\theta}$ and a secondary component $\tilde{p}$ which may, for instance, affect the cost function of the MPC in \cref{eq:MPC}. 
The idea of treating the nominal model as a decision variable has been used, for example, in \cite{gros2019data} and subsequent works by the same authors. Allowing the nominal model to be optimized in this way offers additional flexibility and can enhance the overall performance of the scheme. In \cref{section:equivalence} we focus on an alternative scheme where $\tilde{\theta}$ is set equal to the nominal model at each iteration, aligning the prediction model with the best current estimate of the system dynamics, as we believe this approach is more practical for real-world operation.

While the model in \cref{eq:MPC_prediction_model} captures a general-purpose linearization along a prior trajectory, it can be simplified to reduce computational complexity. For example, one could linearize the nominal dynamics around a fixed state-input pair or use a time-invariant model linearized at the origin. In this paper, we focus on the general case, and refer the reader to \cite[Section VI-A]{zuliani2023bp} for a detailed comparison of these alternatives.

\begin{figure}
\centering
\ifcompile
    \input{graphics/update_loop/update_loop.tex}
\else
    \includegraphics{graphics/update_loop/update_loop_main.pdf}
\fi
\caption{Closed-loop optimization algorithm. Observe that the nominal model $\theta^k$ in iteration $k$ need not match the prediction model $\tilde{\theta}^k$ used by the MPC.}\label{fig:update_loop}
\end{figure}

We assume that system \cref{eq:system} executes a repeated operation starting at time step $t=0$ from a known initial state $x_0$ with a given $y_{-1}$, evolving until $t=T$. 
After each operation, the state is reset to a possibly different $x_0$. We refer to one such operation, spanning $T$ time steps, as an \emph{iteration}.

The performance is measured by a cost function $\mathcal{C}:\R^{(T+1)n_x} \times \R^{Tn_u} \times \R^{n_p}$, which is minimized in expectation over the additive noise and the initial condition. This leads to the following stochastic optimization problem
\begin{align}
\begin{split}
\underset{p,x,u,y}{\text{minimize}} & \quad \mathbb{E}_{v} [\mathcal{C}(x,u,p)]\\
\text{subject to}
& \quad x_{t+1}=f(x_t,u_t,\theta)+w_t,~t\in\Z_{[0,T-1]},\\
& \quad y_t = \operatorname{MPC}(x_t,y_{t-1},p),~t\in\Z_{[0,T-1]},\\
& \quad u_t=\pi(x_t,y_t,p),~u_t \in \mathcal{U},~t\in\Z_{[0,T-1]},\\
& \quad x_t\in \mathcal{X},~p\in \mathcal{P},~t\in\Z_{[0,T]},\\
& \quad v=(w,x_0,y_{-1}),
\end{split}\label{eq:prob_nominal}
\end{align}
where we defined $x=(x_0,\dots,x_T)$ and similarly for $u$ and $y$. 
Note that $y_{-1}$ and $x_0$ may be interdependent, for example if $y_{-1}$ depends on $x_0$. 

%% file: graphics/update_loop/update_loop.tex
\begin{tikzpicture}[thick,every node/.style={font=\footnotesize}]

\node[rectangle,very thick,draw=deployColor, minimum width=2.5cm, minimum height=2.25cm,fill=deployColor!15] (deploy) at (0,-0.2) {};
\node (deploy_text) at (0,0.6) {\textcolor{deployColor}{\textbf{DEPLOY}}};

\node[rectangle,deployColor,draw,minimum height=0.5cm,minimum width=0.8cm] (f) at (0,0) {\footnotesize$f(\cdot)$};
\node[rectangle,deployColor,draw,below of=f,node distance=0.8cm,minimum height=0.5cm,minimum width=0.8cm] (mpc) {\footnotesize MPC};
\draw [-stealth,deployColor] (f.east) to [bend left=50] (mpc.east);
\draw [-stealth,deployColor] (mpc.west) to [bend left=50] (f.west);

\node[rectangle,very thick,draw=updateColor, minimum width=2.5cm, minimum height=0.7cm, fill=updateColor!15,below of =deploy, node distance = 2.25cm] (update) {\textcolor{updateColor}{\textbf{UPDATE}}};

\node[rectangle,very thick,draw=sysIdColor, minimum width=1.7cm, minimum height=0.7cm, fill=sysIdColor!15,right of = deploy, node distance = 3.5cm] (sysid) {\textcolor{sysIdColor}{\textbf{SYS-ID}}};

\node[draw, rectangle,very thick,  draw=sensitivityColor, fill=sensitivityColor!15, below of = sysid, node distance =1.5cm, xshift=-0.45cm, align=center] (gd) {\textcolor{sensitivityColor}{\textbf{Sensitivity}} \\ \textcolor{sensitivityColor}{\textbf{Computation}}};

\draw [-{Stealth[scale=0.6]},very thick,arrowColor] (deploy.east) -- node [pos=0.5, above] {\textcolor{arrowColorText}{$(x^k,u^k)$}} (sysid.west);
\draw [-{Stealth[scale=0.6]},very thick,arrowColor] ([xshift=0.65cm]deploy.east) -- ([xshift=0.65cm,yshift=-.7cm]deploy.east) -- ([xshift=1.3cm,yshift=-.7cm]deploy.east) -- ([xshift=1.3cm,yshift=-1.075cm]deploy.east);
\draw [-{Stealth[scale=0.6]},very thick,arrowColor] (sysid.south) -- node [pos=0.5, right] {\textcolor{arrowColorText}{$(\Theta^k,\theta^k)$}} ([xshift=0.45cm]gd.north);
\draw [-{Stealth[scale=0.6]},very thick,arrowColor] (update.west) -- ([xshift=-0.65cm]update.west) -- node [pos=0.5,left] {\textcolor{arrowColorText}{$p^k=(\tilde{p}^k,\tilde{\theta}^k)$}} ([xshift=-0.65cm]deploy.west) -- (deploy.west);
\draw [-{Stealth[scale=0.6]},very thick, arrowColor] (gd.south) -- ([xshift=1.775cm]update.east) -- node [pos=0.5, below] {\textcolor{arrowColorText}{$J_\Cb^k$}} (update.east);

\end{tikzpicture}

%% file: sources/proposed_method.tex
Solving a constrained, nonconvex, and nonsmooth problem such as \cref{eq:prob_nominal} is very challenging in general.
To simplify the problem we can incorporate the state constraints $x_t\in\mathcal{X}$ in \cref{eq:prob_nominal} into the cost function using a penalty function. 
Let
\begin{align*}
\bar{\mathcal{C}}(p,v) := \mathcal{C}(x(p,v),u(p,v),p) + c_3 \sum_{t=0}^{T} \operatorname*{dist}(x_t(p,v),\mathcal{X}),
\end{align*}
where $x(p,v)=(x_0(p,v),\dots,x_T(p,v))$ and $u(p,v)=(u_0(p,v),\dots,u_{T-1}(p,v))$ are the closed-loop trajectories obtained by fixing $p$ and $v$.
We can reformulate \cref{eq:prob_nominal} as
\begin{align}
\underset{p}{\text{minimize}} & \quad \C(p) := \E_v [\bar{\mathcal{C}}(p,v)],
\label{eq:prob_tractable}
\end{align}
where we set $\pi(x_t, y_t, p) = u_t \in \mathcal{U}$ for all $x_t$, $y_t$, and $p$, and thus omitted the explicit constraint $u_t \in \mathcal{U}$ in \cref{eq:prob_tractable} since $u$ is generally obtained directly from the MPC problem \cref{eq:MPC}, which enforces $u \in \mathcal{U}$ explicitly. 
The state constraint $x_t\in \mathcal{X}$ is incorporated via the penalty term $c_3 \sum_{t=0}^{T} \operatorname*{dist}(x_t(p,v),\mathcal{X})$, for some $c_3>0$, involving the distance to the constraint set $\mathcal{X}$. 
This formulation leads to a design that satisfies state constraints \emph{on average}.

Problem \cref{eq:prob_tractable} cannot be solved directly for two reasons: i) the system parameter $\theta$ is unknown, and ii) the distribution of the disturbance $v$ is unknown. 
To address i), we use the state-input trajectories collected in each iteration to perform system identification and obtain a confidence set $\Theta^k$ for which $\theta\in \Theta^k$ with high probability and a nominal model $\theta^k\in\Theta^k$. 
To address ii), we consider a single value $v^k$ of $v$, reconstructed by collecting state measurements throughout the iteration, and combine it with $\theta^k$ to obtain an estimate of the gradient of the objective function, which we then use to update $p^k$. 
A high-level description of the proposed methodology can be seen in \cref{fig:update_loop}, 
while a more detailed overview is postponed to \cref{alg:main} and \cref{section:algorithm}.
We first outline the two main ingredients of our algorithm (system identification in \cref{section:learning} and gradient estimation in \cref{section:backprop}) and then combine them in \cref{section:algorithm}.

%% file: sources/learning_model.tex
At each iteration $k$, the closed loop system yields noisy measurements $x^k_{t+1}$ of $f(x^k_t,u^k_t)$. Defining $\psi_t^k:=\phi(x_t^k,u_t^k)$ and $z_t^k=x_{t+1}^k-\varphi(x_t^k,u_t^k)$, we have for all $t\in\Z_{[0,T-1]}$
\begin{align}
z_t^k = \psi_t^{k,\top}\theta + w_t^k. \label{eq:sys_id_setting}
\end{align}
Identifying $\theta$ from measurements of the form \cref{eq:sys_id_setting} is a well-studied problem, and efficient algorithms with theoretical guarantees are available. In this work, we adapt the recursive least squares estimator of \cite{abbasi2011improved}
\begin{subequations}\label{eq:LS}
\begin{align}
A^{k+1} & = A^{k} + {\textstyle \sum_{t=0}^{T-1}} \psi_t^k \psi_t^{k,\top}, \label{eq:LS:A}\\
b^{k+1} & = b^{k} + {\textstyle \sum_{t=0}^{T-1}} \psi_t^k z_t^k, \label{eq:LS:B}
\end{align}
with $A^0=\lambda I$ and $b^0=\lambda \theta^0$, where $\lambda>0$ and $\theta^0$ is the best initial guess of $\theta$, and $\theta^k$ is updated via
\begin{align}
\theta^{k+1} & = (A^{k})^{-1} b^k. \label{eq:LS:theta}
\end{align}
\end{subequations}
Following \cite{abbasi2011improved}, we impose the following standard assumption on the disturbance.
\begin{assumption}\label{ass:sys_id}
For each $t$ and $k$, the disturbance samples $w_t^k$ are independent and identically distributed (i.i.d.), almost surely bounded, and $R$-sub-Gaussian for some known $R \geq 0$, that is, for all $\lambda\in\R$, $\E[e^{\lambda w_t^k}] \leq \operatorname*{exp}\left( \frac{\lambda^2R^2}{2} \right)$, where the exponential and the inequality are applied element-wise. 
Moreover, $\|\theta\| \leq S$ for some known $S > 0$.
\end{assumption}
Under \cref{ass:sys_id}, \cite[Theorem 2]{abbasi2011improved} guarantees that given any $\delta \in (0,1)$, with probability at least $1-\delta$, the following holds for all $t\in\Z_{[0,T]}$
\begin{align}
\|\theta^k-\theta\|_{A^k} \leq c_k := R \sqrt{ 2 \log \left( \frac{\det(A^k)^{1/2}}{\det(\lambda I)^{1/2}\delta} \right) } + \lambda^{1/2}S. \label{eq:RLS_ellipsoid}
\end{align}
This bound can be used to construct ellipsoidal sets that are guaranteed to contain $\theta$ with high confidence.

We further assume that the state-input sequence is sufficiently rich to ensure that the uncertainty in the parameter estimate diminishes over iterations. 
\begin{assumption}\label{ass:pe}
For every $k\in\N$, $\sum_{t=0}^{T-1} \psi^k_t \psi^{k,\top}_t \geq \gamma I$ for some $\gamma>0$.
\end{assumption}
\cref{ass:pe} is the \emph{persistency of excitation} (PE) condition, often required in the literature for system identification. However, unlike classical system identification, where excitation conditions are typically required to hold for each time step \cite{johnstone1982exponential}, our PE condition involves the entire iteration $T$. Since $T$ is generally large and stochastic noise $w_t$ is present, satisfying \cref{ass:pe} is not restrictive. 
\rz{We extend our analysis to the case where PE is not satisfied in \cref{section:imperfect_learning}.}

\begin{theorem}\label{thm:sys_id}
Under \cref{ass:sys_id,ass:pe}, for any confidence level $\delta\in(0,1)$, the true parameter $\theta$ belongs to the set $\Theta^k:=\{ \theta: \|\theta-\theta^k\|_{A^k} \leq c_k \}$ for all $k\in\N$ with probability at least $1-\delta$, where $A^k$ and $\theta^k$ are computed in \cref{eq:LS}, and $c_k$ in \cref{eq:RLS_ellipsoid}. 
Moreover, if $\|\psi_t^k\|\leq L_\psi$ for each $t$ and $k$ then $\Theta^k \subseteq \{ \theta: \|\theta-\theta^k\| \leq \tilde{c}_k \}$ where
\begin{align}
\tilde{c}_k = R \sqrt { \frac {  n_\theta \log (1+TkL_\psi^2/n_\theta\lambda) - 2\log(\delta) } { k \gamma } } + \frac{\lambda^{1/2}S}{\sqrt{k \gamma} }, \label{eq:c_k}
\end{align}
and $\lim_{k \to \infty} \tilde{c}_k=0$.
\end{theorem}
\begin{proof}
Equation \cref{eq:RLS_ellipsoid} can be written equivalently as
\begin{align*}
\|\theta^k-\theta\|_{A^k} \leq R \sqrt{\log \left( \frac{\det A^k}{ \lambda^{n_\theta} } \right) - 2 \log \delta } + \lambda^{1/2} S.
\end{align*}
From \cite[Lemma 10]{abbasi2011improved} we have $\operatorname*{det}A^k \leq (\lambda + Tk L_\psi^2/n_\theta)^{n_\theta}$, meaning that
\begin{align*}
\log \left( {\det A^k}/{\lambda^{n_\theta}} \right) \leq n_\theta \log \left( 1 + {Tk L_\psi^2}/{n_\theta \lambda} \right).
\end{align*}
Combining yields
\begin{align*}
\|\theta^k-\theta\|_{A^k} & \leq R \sqrt{ n_\theta \log (1+Tk L_\psi^2/n_\theta \lambda) - 2 \log \delta }  + \lambda^{1/2}S.
\end{align*}
Finally, by \cref{ass:pe} we have $A_k \succeq k \gamma$, and since $\|x\|_A \geq \lambda_\text{min}(A)^{1/2}\|x\|$, we have $\|\theta^k-\theta\| \leq \|\theta^k-\theta\|_{A^k} / \sqrt{k \gamma}$, which combined with the previous equation completes the proof.
\end{proof}
The full system identification procedure is summarized in \cref{alg:sys_id}.
\begin{algorithm}
\caption{System identification.}\label{alg:sys_id}
\begin{algorithmic}[1]
\Input $A^k$, $b^k$, $x^k$, $u^k$, $\delta\in(0,1)$.
\State Compute $\psi_t^k=\phi(x_t^k,u_t^k)$ and $z^k_t=x_{t+1}^k-\varphi(x_t^k,u_t^k)$ for each $t\in\Z_{[0,T-1]}$.
\State Compute $A^k$ and $b^k$ via \cref{eq:LS:A,eq:LS:B}.
\State Update $\theta^{k+1}$ via \cref{eq:LS:theta}.
\State Define $\Theta^{k+1} = \{ \theta\in\R^{n_\theta} : \|\theta^k-\theta\|_{A^k} \leq c_k \}$, where $c_k$ is the RHS of \cref{eq:RLS_ellipsoid}.
\State \Return $\Theta^{k+1},\theta^{k+1}$.
\end{algorithmic}
\end{algorithm}

Our analysis can easily be extended to any system identification technique satisfying %
a bound of the form $\|\theta-\theta^k\|\leq \tilde{c}_k$ for some $\tilde{c}_k=\mathcal{O}(\sqrt{k})$. 
This is beyond the scope of this paper and left as a promising direction for future work.

%% file: sources/backprop.tex
Thanks to the system identification procedure introduced in \cref{section:learning}, 
a nominal model $\theta^k$ is available in each iteration. 
This alone, however, is not enough to fully characterize the Jacobian of the objective of \cref{eq:prob_tractable}, 
which contains an expectation over an unknown distribution. 
Instead of attempting to reconstruct the Jacobian of $\C(p)=\mathbb{E}_v [\bar{\mathcal{C}}(p,v)]$, 
we evaluate the Jacobian of $\bar{\mathcal{C}}$ for a fixed realization of the disturbance $v^k$, observed during the iteration. 
We denote this Jacobian by $\mathcal{J}_{\bar{\mathcal{C}}}^k:=\J_{\bar{\mathcal{C}}}(p^k,v^k,\theta^k)$ to underline the dependency on the nominal model $\theta^k$. 
The Jacobian $\J_{\bar{\mathcal{C}}}^k$ can be obtained by applying the chain rule of differentiation
\begin{multline}
\J_{\bar{\mathcal{C}}}^k = [\mathcal{J}_{\mathcal{C},x}(x^k,u^k,p^k)+\J_{\sum_{t=0}^T\operatorname*{dist}(\cdot_t,\mathcal{X})}(x^k)] \mathcal{J}_x^k \\ + \mathcal{J}_{\mathcal{C},u}(x^k,u^k,p^k) \mathcal{J}_u^k + \mathcal{J}_{\mathcal{C},p}(x^k,u^k,p^k). \label{eq:UL_cost_gradient}
\end{multline}
where $\mathcal{J}_x^k:=\J_{x,p}(p^k,v^k,\theta^k)$ and $\J_u^k=\J_{u,p}(p^k,v^k,\theta^k)$ are obtained recursively as follows
\begin{align}
\J_{x_{t+1}}^k & = \J_{f,x}(x^k_t,u_t^k,\theta^k) \J_{x_t}^k + \J_{f,u}(x^k_t,u_t^k,\theta^k) \J_{u_t}^k, \notag \\
\J_{u_t}^k & = \J_{\pi,x}(x_t^k,y_t^k,p^k) \J_{x_t}^k + \J_{\pi,u}(x_t^k,y_t^k,p^k) \J_{y_t}^k \notag \\ & \quad +\J_{\pi,p}(x_t^k,y_t^k,p^k) \notag \\
\J_{y_t}^k & = \J_{\text{MPC},x}(x_t^k,y_{t-1}^k,p^k)\J_{x_t}^k \notag \\ & \quad + \J_{\text{MPC},y}(x_t^k,y_{t-1}^k,p^k)\J_{y_{t-1}}^k \notag \\ & \quad + \J_{\text{MPC},p}(x_t^k,y_{t-1}^k,p^k) \label{eq:backprop}
\end{align}
initialized with $\J_{x,0}^k=0$ since the initial state $x^k_0$ does not depend on $p$. 
Observe that the closed-loop trajectories $x^k$ and $y^k$ are known, and the only source of error in \cref{eq:UL_cost_gradient} are the Jacobians $\J_x^k$ and $\J_u^k$. 
To compute the conservative Jacobian $\J_{y_t}^k$ of the solution map $y_t$ of \cref{eq:MPC}, we utilize a nonsmooth implicit function theorem outlined in \cref{section:path_diff_opt_prob}.

%% file: sources/algorithm.tex
Our method combines the ideas presented in \cref{section:learning,section:backprop} and it is outlined in \cref{alg:main}. 
Specifically, we apply a projected gradient descent update of the form
\begin{align}
p^{k+1} = \Pi_{\mathcal{Y}^k}[p^k - \alpha_k J_\Cb^k], \label{eq:update}
\end{align}
where $(\alpha_k)_{k\in\N} \subset \R_{>0}$ is a sequence positive stepsizes, 
and $J_\Cb^k\in\J_\Cb^k$ is obtained combining the chain rule in \cref{eq:UL_cost_gradient} and the recursive relation in \cref{eq:backprop} using the nominal model $\theta^k$ obtained using \cref{alg:sys_id}. 
\rz{We stress that the decision variable $p$ in \cref{eq:update} contains both a parameter $\tilde{p}$ affecting the cost of the MPC, and the MPC prediction model $\tilde{\theta}$, and that the gradient descent procedure modifies both parameters. 
The internal MPC model $\tilde{\theta}$ need not be identical to the nominal model $\theta^k$, obtained with the system identification procedure of \cref{alg:sys_id}; 
however, the two are related by the projection onto the set}
\begin{align}
\mathcal{Y}^k:=\{ p = (\tilde{p},\tilde{\theta}) \in \mathcal{P}: \tilde{\theta}\in\Theta^k \},
\label{eq:constraint_set_yk}
\end{align}
representing the set of admissible parameters consistent with the confidence region $\Theta^k$. 
We will present a variation of this scheme where we set $\tilde{\theta}^k=\theta^k$ in \cref{section:equivalence}.

\begin{algorithm}
\caption{Proposed algorithm.}\label{alg:main}
\begin{algorithmic}[1]
\Input $p^0$, $(\alpha_k) _{k\in\N} \subset \R_{>0}$.
\Init $k \gets 0$
\While{\texttt{not\_terminated}}
    \State Sample $v^k$.
    \State Obtain $x^k=x^k(p^k,v^k)$ and $u^k=u(p^k,v^k)$.
    \State Update $\Theta^k$ and $\theta^k$ using \cref{alg:sys_id}.
    \State Compute $J_\Cb^k\in\J_\Cb^k$ using \cref{eq:UL_cost_gradient} and \cref{eq:backprop}.
    \State Update $p^{k+1} = \Pi_{\mathcal{Y}^k}[p^k - \alpha_k J_\Cb^k]$.
\EndWhile
\State \Return $p^k$.
\end{algorithmic}
\end{algorithm}

%% file: sources/convergence.tex
In this section we outline sufficient conditions ensuring that the procedure of \cref{alg:main} converges to a critical point of \cref{eq:prob_tractable}.
First, to ensure existence of the conservative Jacobians $\J_\Cb^k$ of the cost function, we assume the following.
\begin{assumption}\label{ass:local_lip_and_definability}
The functions $\mathcal{C}$, $\text{MPC}$, and $\pi$ are locally Lipschitz and definable. Moreover, $f(\cdot,\cdot,\vartheta)$ is locally Lipschitz and definable for any $\vartheta\in \Theta^k$ and all $k\in\N$.
\end{assumption}
The family of definable functions encompasses a broad range of functions commonly used in control and optimization, making \cref{ass:local_lip_and_definability} a mild requirement. However, establishing the Lipschitz continuity and definability of the $\text{MPC}$ requires more careful consideration. Generally, the solution map of an optimization problem may exhibit discontinuitites, or even be set-valued. In \cite{zuliani2023bp}, we presented sufficient conditions under which the Lipschitz continuity and definability assumption of the MPC map are satisfied for problems formulated as quadratic programs. Additionally, we proposed an efficient algorithm to compute the associated conservative Jacobians at each time step by solving a linear system of equations. \cref{section:path_diff_qp} summarizes the sufficient conditions from \cite{zuliani2023bp}, whereas \cref{section:path_diff_nlp} discusses how these conditions may be extended to more general, nonlinear optimization problems.
\begin{lemma}\label{lemma:path_diff_cost}
Under \cref{ass:local_lip_and_definability} the function $\bar{\mathcal{C}}$ is locally Lipschitz and definable for all $k\geq 0$.
\end{lemma}
\begin{proof}
The functions $x(\cdot)$ and $u(\cdot)$ are obtained as the composition of three mappings, $f$, $\operatorname{MPC}$, and $\pi$, each of which is locally Lipschitz and definable by \cref{ass:local_lip_and_definability}, therefore they are themselves locally Lipschitz and definable by \cite[Exercise 1.11]{coste1999introduction}.
Consequently, the mapping $(p,v)\mapsto\mathcal{C}(x(p,v),u(p,v),p)$ is also locally Lipschitz and definable.
Furthermore, since $\mathcal{X}$ is defined by finitely many linear inequalities, it is definable, and so is the distance function $\operatorname*{dist}(\cdot,\mathcal{X})$ \cite[Exercise 1.15]{coste1999introduction}.
Because the distance function is globally Lipschitz \cite[Proposition 2.4.1]{clarke1990optimization}, the composite mapping $(p,w)\mapsto \sum_{t=0}^{T} \operatorname*{dist}(x_t(p,v),\mathcal{X})$ is therefore locally Lipschitz and definable.
\end{proof}
The following assumption ensures well-posedness of the algorithm.
\begin{assumption}\label{ass:path_diff_set}
For each $k\in\N$, the set $\Theta^k$ is definable. The set $\mathcal{P}$ is also definable.
\end{assumption}
Norm balls, ellipsoids, and polyhedral sets are all definable, thereby making \cref{ass:path_diff_set} a mild requirement. 
The image of a definable set under a definable function is again definable, which gives the following.
\begin{lemma}\label{lemma:Yk_is_definable}
Under \cref{ass:path_diff_set} the set $\mathcal{Y}^k:=\{ p=(\tilde{p},\tilde{\theta})\in \mathcal{P}: \tilde{\theta} \in \Theta^k \}$ is definable.
\end{lemma}
\begin{proof}
The set $\{ p \in \mathcal{P} : \tilde{\theta} = \theta \}$ is definable since $\mathcal{P}$ is definable by \cref{ass:path_diff_set}. 
The set $\mathcal{Y}^k = \{ p \in \mathcal{P}: \tilde{p} \in \Theta^k \}$ can be expressed as $\mathcal{Y}^k = \bigcup_{\theta \in \Theta^k} \{ p \in \mathcal{P} : \tilde{p} = \theta \}$, and since $\Theta^k$ is definable by \cref{ass:path_diff_set}, $\mathcal{Y}^k$ is definable by \cite[Theorem 1.13]{coste1999introduction}.
\end{proof}
To simplify the analysis, we assume the following.
\begin{assumption}\label{ass:P_bounded}
The set $\mathcal{P}$ is bounded.
\end{assumption}
While the boundedness of $\mathcal{P}$ simplifies the analysis, it can be relaxed in favor of a weaker convergence result that would additionally require $\sup_{k\in\N} \|p^k\| < \infty$. For a discussion on boundedness, we refer the reader to \cite[Section 6.1]{davis2020stochastic}.

We further assume that the stepsizes fulfill the following condition, taken from \cite{davis2020stochastic}, needed for convergence.
\begin{assumption}\label{ass:stepsizes}
The stepsizes $( \alpha_k )_{k\in\N} \subset \R_{>0}$ satisfy
\begin{align*}
\sum_{k\in\N} \alpha_k = + \infty, ~~ \sum_{k\in\N} \alpha_k^2 < + \infty.
\end{align*}
\end{assumption}
The condition is standard and it is verified if one chooses, for instance, $\alpha_k = c / k^\gamma$ where $c>0$ and $\gamma\in(0.5,1]$.
We require one last technical assumption involving the sequence of sets $(\mathcal{Y}^k)_{k\in\N}$.
Recall that $N_{\mathcal{Y}^k}$ denotes the \emph{Normal cone} (see \cite[Page 201]{rockafellar2009variational} for a definition) of the set $\mathcal{Y}^k$.
\begin{assumption}\label{ass:outer_semi}
Given any sequence $p^k\to \bar{p}$ and $v^k\in N_{\mathcal{Y}^k}(p^k)$, with $v^k\to \bar{v}$, we have $\bar{v}\in N_\mathcal{Y}(\bar{p})$, where $\mathcal{Y}:=\{ p = (\tilde{p},\tilde{\theta}) \in \mathcal{P}: \tilde{\theta}=\theta \}$.
\end{assumption}
Next, we show that \cref{ass:outer_semi} is verified if each $\Theta^k$ is convex and $\theta \in \Theta^k$ for all $k$, 
an assumption verified by the least-squares algorithm \cref{eq:LS} under \cref{ass:pe}.
\begin{lemma}\label{lemma:convex_satisfies_outersemi}
\cref{ass:outer_semi} is satisfied if $\mathcal{Y}^k$ is convex and $\mathcal{Y} \subset \mathcal{Y}^k$ for all $k$.
\end{lemma}
\begin{proof}
Take any $p^k\to \bar{p}$ and $v^k\in N_{\mathcal{Y}_k}(p^k)$ such that $v_k\to \bar{v}$. 
Then by convexity we have for all $k$ that $\langle v^k, p-p^k \rangle \leq 0$ for all $p\in \mathcal{Y}^k$.
Since by assumption $\mathcal{Y} \subset \mathcal{Y}_k$, we immediately have $\langle v^k, p-p^k \rangle \leq 0$ for all $p\in \mathcal{Y}$,
meaning that $v^k \in N_\mathcal{Y}(p^k)$.
The result then follows from the outer semicontinuity of the normal cone \cite[Proposition 6.6]{rockafellar2009variational}.
\end{proof}

We now state our main result.
\begin{theorem}\label{thm:main}
Under \cref{ass:sys_id,ass:pe,ass:local_lip_and_definability,ass:path_diff_set,ass:stepsizes,ass:P_bounded}, \cref{alg:main} converges to a critical point of \cref{eq:prob_tractable} with arbitrarily high confidence.
\end{theorem}
\begin{proof}
See \cref{section:convergence_proof}.
\end{proof}

%% file: sources/certainty_equivalence.tex
We assume here that the nominal model $\theta^k$ is explicitly incorporated in the dynamics of the MPC at each iteration $k$ (see \cref{fig:update_loop_ce}). 
This approach, known as certainty equivalence, leverages the most plausible model available at each iteration and has a rich history in the control systems literature \cite{bar2003dual}.
In contrast to the approach described in \cref{section:pform,section:pref}, where the model used in the MPC is part of the tunable parameter $p$, certainty equivalence reduces the dimension of $p$, thereby improving computational efficiency and tuning effort.
On the other hand, by constraining the model, we lose degrees of freedom that could potentially enhance performance.

\begin{figure}
\centering
\ifcompile
    \input{graphics/update_loop_ce/update_loop_ce.tex}
\else
    \includegraphics{graphics/update_loop_ce/update_loop_ce_main.pdf}
\fi
\caption{Closed-loop optimization algorithm with certainty equivalence. The MPC prediction model $\tilde{\theta}^k$ now matches nominal model $\theta^k$.}\label{fig:update_loop_ce}
\end{figure}

For each iteration, the closed-loop dynamics are given by
\begin{align}
\begin{split}
x_{t+1}^k&=f(x_t^k,u_t^k,\theta) + w_t^k,\\
y_t^k&=\operatorname{MPC}(x_t^k,y_{t-1}^k,\tilde{p}^k,\theta^k),\\
u_t^k&=\pi(x_t^k,y_t^k,\tilde{p}^k),\\
t&\in\Z_{[0,T]},
\end{split}\label{eq:system_CE}
\end{align}
for a fixed $x_0^k$. The MPC problem is identical to \cref{eq:MPC} with $\tilde{p}^k$ replaced by $\theta^k$
\begin{align*}
\operatorname*{minimize}_{x_{\cdot|t},u_{\cdot|t},\epsilon_{\cdot|t}} & ~~  P_\epsilon(\epsilon_{\cdot|t},\tilde{p}) + \ell_N(x_{N|t},\tilde{p}) + \sum_{j=0}^{N-1} \ell_j(x_{j|t},u_{j|t},\tilde{p}) \notag \\
\text{subject to} & ~~ x_{j+1|t}=A_j(y_{t-1}^k,\theta^k) x_{j|t} + B_j(y_{t-1}^k,\theta^k) u_{j|t} \notag \\ & ~~ + c_j(y_{t-1}^k,\theta^k),~ \forall j \in \Z_{[0,N-1]} \notag \\ 
& ~~ H_x x_{j|t} \leq h_x + \epsilon_{j|t} ,~\epsilon_{j|t} \geq 0,~ \forall j\in\Z_{[0,N]}, \notag \\
& ~~ H_u u_{j|t} \leq h_u,~ \forall j\in\Z_{[0,N-1]}, \notag \\
& ~~ x_{0|t} = x_t^k,
\end{align*}
where the prediction model in \cref{eq:MPC_prediction_model} now becomes
\begin{align*}
A_j(y_{t-1},\theta^k) & = \nabla_x f(x_{j+1|t-1},u_{j+1|t-1},\theta^k),\\
B_j(y_{t-1},\theta^k) & = \nabla_u f(x_{j+1|t-1},u_{j+1|t-1},\theta^k),\\
c_j(y_{t-1},\theta^k) & = f(x_{j+1|t-1},u_{j+1|t-1},\theta^k) \\ & \quad - A_j(y_{t-1},\theta^k)x_{j+1|t-1} - B_j(y_{t-1},\theta^k)u_{j+1|t-1} ,
\end{align*}
For simplicity we denote $x(p^k,v^k)=x(\tilde{p}^k,\theta^k,v^k)$ and similarly $u(p^k,v^k)=u(\tilde{p}^k,\theta^k,v^k)$, 
and we define $\tilde{P}=\{ \tilde{p} : \exists \tilde{\theta} , ~ (\tilde{p},\tilde{\theta}) \in \mathcal{P} \}$.
The certainty equivalence procedure is summarized in \cref{alg:main_CE}.

\begin{algorithm}
\caption{Certainty equivalence algorithm.}\label{alg:main_CE}
\begin{algorithmic}[1]
\Input $p^0$, $(\alpha_k )_{k\in\N} \subset \R_{>0}$.
\Init $k \gets 0$
\While{\texttt{not\_terminated}}
    \State Sample $v^k$.
    \State Obtain $x^k = x(\tilde{p}^k,\theta^k,v^k)$ and $u^k=u(\tilde{p}^k,\theta^k,v^k)$.
    \State Update $\theta^k$ using \cref{alg:sys_id}.
    \State Compute $J_\Cb^k \in \J_\Cb^k$ using \cref{eq:UL_cost_gradient} and \cref{eq:backprop}.
    \State Update $\tilde{p}^{k+1} \gets \Pi_{\tilde{\mathcal{P}}}[\tilde{p}^k - \alpha_k J_\Cb^k]$.
\EndWhile
\State \Return $\tilde{p}^k$, $\theta^k$.
\end{algorithmic}
\end{algorithm}

\begin{corollary}
Under \cref{ass:sys_id,ass:pe,ass:local_lip_and_definability,ass:path_diff_set,ass:stepsizes,ass:P_bounded}, \cref{alg:main_CE} converges to a critical point of \cref{eq:prob_tractable} with arbitrarily high confidence if \cref{alg:sys_id} is used to identify the system.
\end{corollary}
\begin{proof}
We can see that \cref{lemma:bound_gradient} continues to hold in the certainty equivalence setting by recognizing the local Lipschitz continuity of $J_\Cb^k(\theta^k)$ in $\theta^k$, and leveraging the boundedness of $\mathcal{P}$. Then, the result immediately follows from the proof of \cref{thm:main}.
\end{proof}

%% file: graphics/update_loop_ce/update_loop_ce.tex
\begin{tikzpicture}[thick,every node/.style={font=\footnotesize}]

\node[rectangle,very thick,draw=deployColor, minimum width=2.5cm, minimum height=2.25cm,fill=deployColor!15] (deploy) at (0,-0.2) {};
\node (deploy_text) at (0,0.6) {\textcolor{deployColor}{\textbf{DEPLOY}}};

\node[rectangle,deployColor,draw,minimum height=0.5cm,minimum width=0.8cm] (f) at (0,0) {\footnotesize$f(\cdot)$};
\node[rectangle,deployColor,draw,below of=f,node distance=0.8cm,minimum height=0.5cm,minimum width=0.8cm] (mpc) {\footnotesize MPC};
\draw [-stealth,deployColor] (f.east) to [bend left=50] (mpc.east);
\draw [-stealth,deployColor] (mpc.west) to [bend left=50] (f.west);

\node[rectangle,very thick,draw=updateColor, minimum width=2.5cm, minimum height=0.7cm, fill=updateColor!15,below of =deploy, node distance = 3cm] (update) {\textcolor{updateColor}{\textbf{UPDATE}}};

\node[rectangle,very thick,draw=sysIdColor, minimum width=1.7cm, minimum height=0.7cm, fill=sysIdColor!15,right of = deploy, node distance = 4cm] (sysid) {\textcolor{sysIdColor}{\textbf{SYS-ID}}};

\node[draw, rectangle,very thick,  draw=sensitivityColor, fill=sensitivityColor!15, below of = sysid, node distance =1.5cm, xshift=-1.5cm, align=center] (gd) {\textcolor{sensitivityColor}{\textbf{Sensitivity}} \\ \textcolor{sensitivityColor}{\textbf{Computation}}};

\draw [-{Stealth[scale=0.6]},very thick,arrowColor] (deploy.east) -- node [pos=0.5, above] {\textcolor{arrowColorText}{$(x^k,u^k)$}} (sysid.west);
\draw [-{Stealth[scale=0.6]},very thick,arrowColor] ([xshift=0.85cm]deploy.east) -- ([xshift=0.85cm,yshift=-1.07cm]deploy.east);
\draw [-{Stealth[scale=0.6]},very thick,newArrowColor] ([yshift=-0.325cm]sysid.south) -- node [pos=0.42, right] {\textcolor{newArrowColor}{$\theta^k$}} ([yshift=-2.775cm]sysid.south) -- ([yshift=-0.15cm]update.east);
\draw [-{Stealth[scale=0.6]},very thick,arrowColor] (update.west) -- ([xshift=-0.65cm]update.west) -- node [pos=0.4,right] {\textcolor{arrowColorText}{$p^k=(\tilde{p}^k,\textcolor{newArrowColor}{\theta^k})$}} ([xshift=-0.65cm]deploy.west) -- (deploy.west);
\draw [-{Stealth[scale=0.6]},very thick, arrowColor] (sysid.south) -- ([yshift=-0.325cm]sysid.south) -- ([yshift=-0.325cm,xshift=-1.15cm]sysid.south) -- ([yshift=-0.7cm,xshift=-1.15cm]sysid.south);
\draw [-{Stealth[scale=0.6]},very thick, arrowColor] (gd.south) -- ([yshift=0.15cm,xshift=1.225cm]update.east) -- node [pos=0.5, above] {\textcolor{arrowColorText}{$J_\Cb^k$}} ([yshift=0.15cm]update.east);

\end{tikzpicture}

%% file: sources/imperfect_learning.tex
Without \cref{ass:pe}, the system identification procedure in \cref{eq:LS} is no longer guaranteed to yield an asymptotically exact estimate of the true dynamics,
and the radius $c_k$ of the confidence region $\Theta^k$ may not shrink to zero.
This introduces a challenge for the subsequent analysis,
as imperfect model knowledge can lead to non-vanishing error terms.
For the rest of this section, we assume that system is noiseless meaning that $\C(p)=\bar{\mathcal{C}}(p)$. 
We briefly outline a possible extension to the noisy case at the end of this section.

In the absence of \cref{ass:pe}, 
one can design a parameter vector $\bar{p}$ that minimizes the distance between $0$ and $\J_\C(\bar{p})$,
while simultaneously providing an upper bound on this (generally unknown) distance.
To achieve this, we apply \cref{alg:main_CE} (or equivalently, \cref{alg:main}) without updating the nominal model,
that is, by fixing $\theta^k \equiv \theta^0$ and exploit the fact that $\theta \in \Theta^0$, to construct a set of candidate upper bounds
\begin{align} \label{eq:imperfect_learning:Smax}
S_{\text{max}}^M := \left\{{\textstyle\max_{J\in \J_\Cb(\bar{p},\theta_i)}} \|J\| : \theta_i \in \Theta^0, ~ i\in\Z_{[0,M]}\right\},
\end{align}
where each $\theta_i$ is sampled from a known probability distribution $\mathbb{P}_\theta$ supported on $\Theta^0$, and $M$ is the number of samples.

Given a confidence level $\beta \in (0,1)$,
the set $S_{\text{max}}^M$ can be used to compute an upper bound $J_\text{max}^{\epsilon,\beta}$ that satisfies,
with probability at least $1 - \beta$,
\begin{align}
\mathbb{P}\big[ \| \J_\C(\bar{p},\theta) \| > J_\text{max}^{\epsilon,\beta} \big] \leq \epsilon,
\end{align}
for a prescribed violation risk $\epsilon \in (0,1)$.

\begin{theorem}\label{thm:imperfect_learning}
For $\beta\in(0,1)$ and $\epsilon\in(0,1)$, if
\begin{align}
\binom{k+n_\theta-1}{k} \sum_{i=0}^{k+n_\theta-1} \binom{M}{i} \epsilon^i(1-\epsilon)^{M-i} \leq \beta, \label{eq:imperfect:scenario_condition}
\end{align}
then, with probability at least $1-\beta$,
\begin{align*}
\mathbb{P}[\| \J_\C(\bar{p},\theta) \| > J_\text{max}^{\epsilon,\beta}] \leq \epsilon,
\end{align*}
where $J_\text{max}^{\epsilon,\beta}$ is the $k$-th largest element of $S_{\text{max}}^M$.
\end{theorem}
\begin{proof}
Consider the following scenario program
\begin{align}
\begin{split}
\operatorname*{minimize}_{x} & \quad x\\
\text{subject to} & \quad x\in\mathcal{X}_{\theta_i}, ~~ i\in[0,M],
\end{split} \label{eq:imperfect:scenario_program}
\end{align}
where $\mathcal{X}_{\theta_i}:=\{ x: \|\J_\C(\bar{p},\theta_i)\| \leq x \}$.
Problem \cref{eq:imperfect:scenario_program} is always feasible,
satisfying \cite[Assumption 2.1]{campi2011sampling}.
Moreover, by removing the $k$ largest elements of $S_\text{max}^M$,
that is, the $k$ largest values of $\|\J_\C(\bar{p},\theta_i)\|$,
the corresponding solution $x$ that solves \cref{eq:imperfect:scenario_program} almost surely violates the removed constraints,
thereby verifying \cite[Assumption 2.2]{campi2011sampling}.
The result then follows directly from \cite[Theorem 2.1]{campi2011sampling}.
\end{proof}
Once an acceptable violation risk $\epsilon$ has been set, and given a number of data points $M$, we can test the condition in \cref{eq:imperfect:scenario_condition} for increasing values of $k$ to understand how many values of $S_\text{max}^M$ can be discarded while fulfilling the required violation probability.
This is summarized in \cref{alg:upper_bound}.
The tradeoff between violation risk and magnitude of the bound is empirically studied in \cref{section:simulation:lateral_control}.

\begin{algorithm}
\caption{Norm upper bound determination.} \label{alg:upper_bound}
\begin{algorithmic}[1]
\Require $\Theta$, $\bar{p}$, $\epsilon,\beta\in(0,1)$, $M>0$.
\Init $k\gets 1$
\State Sample $\theta_i\in \Theta$ for $i=1,\dots,M$.
\State Form $S_\text{max}^M$ as in \cref{eq:imperfect_learning:Smax}.
\State $k_\text{max}=\operatorname*{argmax}_{k \geq 0} \{ k: \text{\cref{eq:imperfect:scenario_condition} holds} \}$
\State \Return $J_\text{max}^{\epsilon,\beta} = $ $k_\text{max}$-largest element of $S_\text{max}^M$.
\end{algorithmic}
\end{algorithm}

If the system is noisy, the upper bound computation in \cref{eq:imperfect_learning:Smax} can be modified to
\begin{align*}
S_{\text{max}}^M := \left\{ \frac{1}{M} \sum_{j=1}^M \max_{J\in \J_\Cb(\bar{p},\theta_i,v_j)} \|J\| : \theta_i \in \Theta^0, ~ i\in\Z_{[0,M]}\right\},
\end{align*}
where each $v_j=(w_j,x_{0,j},y_{-1,j})$ is obtained from the real system \cref{eq:system} with $p=\bar{p}$. 
Unlike \cref{eq:imperfect_learning:Smax}, here we require multiple rollouts of the trained MPC on the real system to collect enough closed-loop trajectories and construct a sample average approximation of the expected cost. 
One could then modify \cref{thm:imperfect_learning} to account for the approximation leveraging finite-sample results that apply to sample average approximation. 
This is beyond the scope of this paper and is left as a direction for future work.


%% file: sources/simulation_results.tex
\subsection{Random Linear Systems}\label{section:sim:random_linear}

We start by deploying Algorithms \ref{alg:main} and \ref{alg:main_CE} on a set of randomly generated linear systems.
For each simulated system, we start by randomly sampling real poles 
(uniformly at random from the interval $[-0.1,0.1]$) 
of a $4$-dimensional continuous-time single input linear system in controllable canonical form, 
and then discretize the resulting system (exactly) with a sampling time of $0.15\,\text{s}$. 
The system dynamics are given by
\begin{align*}
x_{t+1} = A x_t + B(u_t + w_t),
\end{align*}
where $t\in\Z_{[0,50]}$, $x_0$ is uniformly sampled for each iteration from a ball of radius $1.5$,
and $w_t\in\R$ is uniformly sampled at each time-step such that $|w_t| \leq 0.1$.
We then set $\theta$ equal to the entries of the $A$ and $B$ matrices and randomly generate a nominal value $\theta^0$ of $\theta$ such that $\|\theta^0-\theta\|=0.3 \|\theta\|$.
The MPC problem is given by \cref{eq:MPC} with $N=5$ and
\begin{align*}
\ell_t(x_t,u_t,p) =
\begin{cases}
\|x_t\|_Q^2 + \|u_t\|_R^2, & \text{if } t \neq N,\\
\|x_t\|_P^2, & \text{if } t = N,
\end{cases}
\end{align*}
where $p_2\in\R$, $Q=p_1p_1^\top+10^{-8}I$, $R=p_2^2+10^{-8}$, and $P=p_3p_3^\top+10^{-8}I$, where $p_1$ and $p_3$ are lower-triangular matrices, and $p=\operatorname*{col}(p_1,p_2,p_3,\tilde{\theta})$. We do not consider state constraints and only enforce input constraints  $u_t\in[-1,1]$.
The upper-level cost is given by
\begin{align*}
\mathcal{C}(x,u) = \sum_{t=0}^{T} \|x_t\|^2_{\mathcal{Q}} + \sum_{t=0}^{T-1} \|u_t\|^2_{\mathcal{R}},
\end{align*}
where $\mathcal{Q}=10 I$ and $\mathcal{R}=I$.
We choose $p^0$ to ensure that the initial values of $Q$ and $R$ in the MPC match $\mathcal{Q}$ and $\mathcal{R}$, respectively, $P=Q$, and $\tilde{\theta}=\theta^0$.

We run \cref{alg:main} and \cref{alg:main_CE} for $500$ iterations using gradient descent with stepsizes $\alpha_k = 10^{-3} / k^{0.8}$.
To avoid large overshoots, we clip the norm of the gradient to 50.
This does not hinder the convergence properties as long as the clipping happens finitely many times.

Our tuning algorithm significantly improves the performance of the MPC on the training and testing sets (containing $500$ unseen samples), as can be seen in the two left-most columns of \cref{tab:sim:random_linear_results}.
For a more graphical representation, we took the first 10 experiments in the training set and plotted each as a separate column in \cref{fig:sim:tab_no_ce}, where the yellow, blue, and gray segments of each column represent the cost attained by the omniscient, trained, and untrained algorithm (averaged across all $500$ samples).
In \cref{fig:sim:tab_no_ce} and \cref{tab:sim:random_linear_results}, the suboptimality is computed against an omniscient controller with foreknowledge of $\theta$ and $\{w_t\}$,
whose performance is unattainable with a feedback controller.
In addition, as shown in \cref{tab:sim:random_linear_results}, the tuned MPC significantly outperforms a nominal MPC scheme where the terminal cost $P$ is chosen as the solution of the Riccati equation (using the true costs $\mathcal{Q}$ and $\mathcal{R}$, and the nominal model $A$ and $B$ obtained after running system identification for all $500$ iterations), denoted with DARE in the Table.
Note that the DARE controller attains different costs in the CE and no CE cases because of the different data that the schemes collect, used for the identification of the model.

\begin{table}
\centering
\caption{Relative suboptimality of various controllers to the omniscient controller on training set ($\mathcal{S}_\text{tr}$) and testing set ($\mathcal{S}_\text{te}$).}\label{tab:sim:random_linear_results}
\begin{tabular}{lcccc}
\toprule
& \multicolumn{2}{c}{\textbf{Alg. \ref{alg:main} (no CE)}} & \multicolumn{2}{c}{\textbf{Alg. \ref{alg:main_CE} (CE)}} \\
\cmidrule(lr){2-3} \cmidrule(lr){4-5}
& \emph{Mean} & \emph{Variance} & \emph{Mean} & \emph{Variance} \\
Nominal ($\mathcal{S}_\text{tr}$)    & 8.15 & 12.25 & 8.15 & 12.25 \\
Tuned ($\mathcal{S}_\text{tr}$)      & 0.30 &  0.00 & 0.11 &  0.00 \\
Tuned ($\mathcal{S}_\text{te}$)      & 0.30 &  0.00 & 0.11 &  0.00 \\
DARE ($\mathcal{S}_\text{te}$)       & 1.93 &  0.66 & 1.95 &  0.69 \\
\bottomrule
\end{tabular}
\end{table}

\begin{figure}
\centering
\ifcompile
    \input{graphics/no_ce_tab/no_ce_tab.tex}
\else
    \includegraphics{graphics/no_ce_tab/no_ce_tab_main.pdf}
\fi
\caption{Cost ($y$-axis) and relative suboptimality (typed number in columns) of the MPC trained with \cref{alg:main} (blue) against the performance of the omniscient controller (yellow) and untrained algorithm (gray), on ten randomly generated linear systems. The $+xx.xx$ number in the grey section of any column indicates that the untrained controller performed $xx.xx$ times worse than the trained controller in that experiment. The $-0.yy$ number in the yellow section indicates that the tuned controller performed $0.yy$ times worse than the omniscient controller.}\label{fig:sim:tab_no_ce}
\end{figure}

In \cref{fig:sim:tab_ce} and the two right-most columns of \cref{tab:sim:random_linear_results} we repeat the same experiments using \cref{alg:main_CE} (thus, the results for the untrained algorithm remain the same).
The CE variant consistently outperforms \cref{alg:main} across nearly all scenarios (with an average suboptimality of $11\%$ instead of $30\%$).
This improvement stems from the reduced complexity of the CE formulation: since the optimization parameter no longer includes the prediction model, the decision space is significantly smaller, and the CE algorithm tends to locate local minimizers more effectively.

\begin{figure}
\centering
\ifcompile
    \input{graphics/ce_tab/ce_tab.tex}
\else
    \includegraphics{graphics/ce_tab/ce_tab_main.pdf}
\fi
\caption{Cost ($y$-axis) and relative suboptimality (typed number in columns) of the MPC trained with \cref{alg:main_CE} (blue) against the performance of the omniscient controller (yellow) and untrained algorithm (gray), on ten randomly generated linear systems.}\label{fig:sim:tab_ce}
\end{figure}

\subsection{Nonlinear Quadcopter}

Next, we deploy the CE variant on a $12$-dimensional nonlinear quadcopter taken from \cite{abdulkareem2022modeling}.
The state of the system is given by the position and velocity in the inertial frame, the Euler angles, and the angular velocity in the body frame.
Following \cite{abdulkareem2022modeling}, we denote with $p_x$, $p_y$, $p_z$, $v_x$, $v_y$, $v_z$ the position and velocity along the three axes,
and with $\phi$, $\vartheta$, $\psi$, $p$, $q$, $r$ the Euler angles and the angular rates in the body frame.
The control inputs of the system are the rotation speeds $\omega_1$, $\omega_2$, $\omega_3$, $\omega_4$, of each of the four propellers.
The thrust $u_T$ produced by the propellers is aligned with the $z$-axis of the body frame (i.e., pointing upwards with respect to the orientation of the drone) and given by $u_T=\operatorname*{col}(0,0,k_t \sum_{i=1}^{4} \omega_i^2)$, where $k_t$ is the thrust coefficient.
The propellers additionally produce a torque on the drone, given by
\begin{align*}
\tau = \begin{bmatrix}
\tau_\phi \\ \tau_\vartheta \\ \tau_\psi
\end{bmatrix} = 
\begin{bmatrix}
l k_t (\omega_4^2-\omega_2^2)\\ lk_t (\omega_3^2-\omega_1^2) \\ k_b (\omega_2^2-\omega_1^2+\omega_3^2-\omega_4^2)
\end{bmatrix},
\end{align*}
where $k_b$ is the drag coefficient and $l$ is the lateral length of the drone.
The linear acceleration is
\begin{align*}
\dot{v} & = R_z(\phi,\vartheta,\psi) \frac{u_T}{m} - k_d v - g,
\end{align*}
where $v=\operatorname*{col}(v_x,v_y,v_z)$, $g=\operatorname*{col}(0,0,9.8)$ is the gravitational acceleration, $k_d$ is the air resistance coefficient, $m$ is the mass of the drone, and $R_z(\phi,\vartheta,\psi) u_T$ provides the thrust $u_T$ in the inertial frame, where $R_z(\phi,\vartheta,\psi)$ is the third column of the rotation matrix from the body frame to the inertial frame
\begin{align*}
R_z(\phi,\vartheta,\psi) = \begin{bmatrix}
\cos{\psi}\sin{\vartheta}\cos{\phi} + \sin{\psi}\sin {\phi} \\
\sin{\psi}\sin{\vartheta}\cos{\phi} - \cos{\psi}\sin {\phi} \\
\cos {\vartheta} \cos (\phi) 
\end{bmatrix}.
\end{align*}
The rotational dynamics are given by
\begin{align*}
\dot{\omega}_B = I^{-1} (- \omega_B \times I \omega_B - J_r \omega_B \omega_r + \tau),
\end{align*}
where $\omega_B=\operatorname*{col}(p,q,r)$ is the angular velocity in the body frame, $I=\operatorname*{diag}(I_{xx},I_{yy},I_{zz})$ is the inertia matrix (under the assumption that the drone is axis symmetrical), $J_r$ is a constant, and $\omega_r=\omega_2-\omega_1+\omega_4-\omega_3$.

To obtain \cref{eq:system}, we can set $\theta=(\theta_i)_{i=0}^{11}$ and write
\begin{align*}
\dot{x} = \begin{bmatrix}
0_3\\
\theta_0 [\sum_{i=1}^{4} R_z\omega_i^2 - g] - \operatorname*{col}(\theta_1 v_x , \theta_2 v_y, \theta_3 v_z) \\
\omega_\eta(\phi,\vartheta,\psi) \omega_B\\
\theta_4 qr - \theta_5 q \omega_r + \theta_6 (\omega_4^2-\omega_2^2)\\
\theta_7 pr + \theta_8 p \omega_r + \theta_9 (\omega_3^2-\omega_1^2)\\
\theta_{10} pq + \theta_{11} (\omega_2^2-\omega_1^2+\omega_3^2-\omega_4^2)
\end{bmatrix}+
\begin{bmatrix}
v \\ 0_3 \\ 0_3 \\ 0 \\ 0 \\ 0
\end{bmatrix}
\end{align*}
where $x=\operatorname*{col}(p_x,p_y,p_z,v,\phi,\vartheta,\psi,\omega_B)$ and
\begin{align*}
\omega_\eta(\phi,\vartheta,\psi) = \begin{bmatrix}
1 & 0 & -\sin{\vartheta}\\
0 & \cos{\phi} & -\sin{\phi}\cos{\vartheta}\\
0 & -\sin{\phi} & \cos{\phi}\cos{\vartheta}
\end{bmatrix}
\end{align*}
relates Euler angle rates to angular velocity in the body frame.
We discretize the dynamics using Euler's forward scheme and a sampling time of $0.1\,\text{s}$.
We assume that $\theta$ is unknown and randomly choose $\theta^0$ such that $\|\theta^0-\theta\|=0.5 \|\theta\|$.
The true values of $\theta$ are taken from \cite[Table 1]{abdulkareem2022modeling}.

For the MPC we choose $N=12$ and
\begin{align*}
\ell_t(x_t,u_t,p) =
\begin{cases}
\|x_t-x_\text{ref}\|_Q^2 + \|u_t-u_\text{ref}\|_R^2, & \text{if } t \neq N,\\
\|x_t-x_\text{ref}\|_P^2, & \text{if } t = N,
\end{cases}
\end{align*}
where $x_\text{ref}=(-6,-3.5,0,0_9)$ and $u_\text{ref}$ is the input required to maintain the drone at a hovering state (which we assume to be available).
We enforce the constraints $\omega_i\in [0,630]$, $v_x,v_y,v_z\in[-2,2]$, $\phi,\vartheta,\psi\in [-\pi/4,\pi/4]$, and $p,q,r\in[-\pi/8,\pi/8]$.
The parameter $p$ is chosen as $p=(p_1,p_2,p_3)$ where $p_1$, $p_2$, and $p_3$ are as in \cref{section:sim:random_linear}. 
The upper-level horizon is set to $T=100$, and the upper-level cost is 
\begin{align*}
\mathcal{C}(x,u) & = \sum_{t=0}^{T} c_3\operatorname*{dist}(x_t,\mathcal{X})^2 + \|x_t-x_\text{ref}\|^2_{\mathcal{Q}}+ \sum_{t=0}^{T-1} \|u_t-u_\text{ref}\|^2_{\mathcal{R}}
\end{align*}
where $c_3=200$ (chosen empirically), $\mathcal{Q}=10I$ and $\mathcal{R}=I$.

We run \cref{alg:main_CE} for $200$ iterations with $\rho=5\cdot 10^{-5}$ and $\eta=0.6$. A comparison between the angle and position trajectories of the MPC at iteration $0$ and iteration $200$ can be seen in \cref{fig:quad_state_trajectories}. Observe how the trajectories have effectively converged to the ones obtained with an omniscient controller that knows $\theta$ and uses a large prediction horizon $T$ (denoted in red in the top three plots, and in the lines with square markers in the bottom plot).
From \cref{fig:quad_convergence}, we can see that the RLS procedure (using $\lambda=10^{-8}$) successfully learns the model (achieving an estimation error smaller than $1\%$ in about $40$ iterations), whereas the closed-loop cost decreases and approaches the best achievable.
The convergence behavior is consistent with the one obtained if the model was known at iteration 0 (orange line in the plot above) thus validating our theoretical findings.

To further assess the capabilities of the proposed algorithm, we repeated the experiment by adding to the upper-level cost the nonconvex term $10^{-7}\sum_{t=0}^T \sum_{i=1}^4 \omega_i^3$, which penalizes the energy consumption associated with propeller usage.
Unlike the previous example, the MPC formulation in this case does not account for this additional energy penalty in its cost function.
Nevertheless, our algorithm achieves a closed-loop cost of 10594.67, compared to the optimal value of 10589.24, demonstrating that the proposed tuning method can produce effective MPC controllers even if the MPC cost function does not exactly match the upper-level objective.
The closed-loop trajectories of the three Euler angles after training are shown in \cref{fig:quad_energy}.

\begin{figure*}
\centering
\ifcompile
    \input{graphics/quad2/quad2.tex}
\else
    \includegraphics{graphics/quad2/quad2_main.pdf}
\fi
\caption{Angle (top) and position (bottom) trajectories for the untrained / trained / best algorithm. In the position plot, the untrained trajectories are the dashed lines, the trained ones are solid lines, and the optimal ones are the dash-dotted lines with square markers added since they fall mostly below the solid lines.}\label{fig:quad_state_trajectories}
\end{figure*}

\begin{figure}
\centering
\ifcompile
    \input{graphics/quad1/quad1.tex}
\else
    \includegraphics{graphics/quad1/quad1_main.pdf}
\fi
\caption{Convergence of cost (top) and estimation error (bottom) along iterations.}\label{fig:quad_convergence}
\end{figure}

\begin{figure*}
\centering
\ifcompile
    \input{graphics/quad3/quad3.tex}
\else
    \includegraphics{graphics/quad3/quad3_main.pdf}
\fi
\caption{Euler angles of tuned MPC with (gray solid line) and without (red dash-dotted line) energy penalization.}\label{fig:quad_energy}
\end{figure*}

\subsection{Lateral Control of an Autonomous Vehicle} \label{section:simulation:lateral_control}

As a last simulation example we consider an autonomous car racing on a curvilinear track.
We use the bicycle model proposed in \cite[Section 4.12]{snider2009automatic}, 
where the effect of the path curvature is treated as an external disturbance.
We assume that the longitudinal velocity is controlled separately and set it as a constant through the entire track 
and focus on the linearized lateral dynamics
\begin{align}
\begin{bmatrix} \dot{e}_\text{cg} \\ \ddot{e}_\text{cg} \\ \dot{\theta}_\text{e} \\ \ddot{\theta}_\text{e} \end{bmatrix} = 
\begin{bmatrix} 0 & 1 & 0 & 0 \\ 0 & a_1 & a_2 & a_3 \\ 0 & 0 & 0 & 1 \\ 0 & a_4 & a_5 & a_6 \end{bmatrix} 
\begin{bmatrix} e_\text{cg} \\ \dot{e}_\text{cg} \\ \theta_\text{e} \\ \dot{\theta}_\text{e} \end{bmatrix}
+ \begin{bmatrix} 0 \\ b_1 \\ 0 \\ b_2 \end{bmatrix} u + \begin{bmatrix} 0 \\ a_3-v_x \\ 0 \\ a_6 \end{bmatrix} r,
\label{eq:simulation:car_model}
\end{align}
where $e_\text{cg}$ is the lateral tracking error, that is, 
the orthogonal distance between the center of gravity of the car and the center of the track,
and $\theta_\text{e}$ is the orientation error, that is, 
the difference between the heading of the car and the tangent direction to the track at the point on the track that is closest to the center of gravity of the car.
The input $u$ denotes the steering angle and is constrained with the interval $u\in[-\frac{\pi}{5},\frac{\pi}{5}]$.
The effect of the path curvature is represented by $r(t)=\kappa(t)v_x$, 
where $v_x$ is the (constant) longitudinal velocity, 
and $\kappa(t)$ is the instantaneous curvature of the path.
We impose the constraints
\begin{align*}
e_\text{cg} \in [-1,1],~ \dot{e}_\text{cg} \! \in [-5,5], ~ \theta_\text{e} \in [-1,1], ~ \dot{\theta}_\text{e} \in [-2.75,2.75],
\end{align*}
and consider the following parameters (taken from \cite{snider2009automatic}),
$a_1=-27.280$, $a_2=272.798$, $a_3=a_4=a_5=0$, $a_6=-29.388$, $b_1=136.399$, $b_2=126.129$,
with $v_x=\SI{10}{\meter/\second}$.

The continuous-time dynamics in \cref{eq:simulation:car_model} are discretized with a sampling period of $0.01\,\text{s}$ using Euler's forward scheme.
The initial condition is $e_\text{cg}(0)=0.75$, $\dot{e}_\text{cg}(0)=\theta_\text{e}(0)=\dot{\theta}_\text{e}(0)=0$.

The system parameters cannot be identified using the RLS technique introduced in \cref{section:learning} due to the presence of the unknown disturbance $r(t)$,
which is neither stochastic nor zero mean.
While one could, in principle, apply iterative disturbance estimation techniques combined with system identification to learn both the system parameters and $r(t)$, 
we simply deploy the method of \cref{section:imperfect_learning}.

We assume that the parameter $\theta=(a_1,\dots,a_6,b_1,b_2)$ describing the system dynamics is unknown and that the best estimate available is $\theta^0=\alpha_\theta (\mathds{1} + e_\theta) \odot \theta$, where $e_\theta$ is a randomly generated vector with unit norm, 
$\alpha_\theta>0$ is a scalar, 
and $\odot$ denotes componentwise product.

The MPC is formulated as in \cref{eq:MPC}, 
with $A_j \equiv \bar{A}$, 
$B_j \equiv \bar{B}$,
and $c_j \equiv 0$,
where $\bar{A}$, $\bar{B}$ are obtained from \cref{eq:simulation:car_model} using the nominal parameter $\theta^0$.
We choose the horizon to be $N=5$, and use the same cost parameterization as in \cref{section:sim:random_linear}, with the addition of a slack penalty $P_\epsilon$, chosen as in \cref{eq:problem_formulation:mpc_slack_penalty} with $c_1=c_2=25$.
For the upper-level problem we choose the cost $\mathcal{C}(x,u)=\sum_{t=0}^{T} \|x_t\|^2 + \sum_{t=0}^{T-1} 10^{-6}\|u_t\|^2$, and $c_3=100$.

\begin{figure}[ht]
\centering
\ifcompile
    \input{graphics/car_path/car_path.tex}
\else
    \includegraphics{graphics/car_path/car_path_main.pdf}
\fi
\caption{The curvilinear track considered in \cref{section:simulation:lateral_control} (in gray), the trajectory obtained after optimization (solid orange line), the initial trajectory (blue), and the optimal trajectory (dashed red).}
\label{fig:simulation:lateral_control:path}
\end{figure}

\cref{fig:simulation:lateral_control:path} shows the curvilinear track we considered, which was generated by interpolating waypoints using splines and reports trajectories associated to different controllers.
Observe how the trajectory of the trained controller (in red) is significantly closer to the trajectory obtained by the omniscient controller (in orange) compared to the untrained one (in blue).
The best achievable trajectory was obtained by solving a trajectory optimization problem using the true cost and the true dynamics,
assuming foreknowledge of $r(t)$.

\begin{table}[hb]
\caption{Closed-loop costs of different controllers.}
\label{table:simulation:laterl_control:closed_loop_costs}
\centering
\begin{tabular}{ l c c c}
\toprule
\textbf{Controller} & Trained & Untrained & Best Cost \\
\midrule
\textbf{Cost} & 247.323 & 379.762 & 209.982 \\
\bottomrule
\end{tabular}
\end{table}

The closed-loop costs of the three controllers are reported in \cref{table:simulation:laterl_control:closed_loop_costs}.
Training ensures a $0.35$ reduction ($35\%$) in the closed-loop cost,
achieving a suboptimality of $0.15$ ($15\%$).
\cref{fig:simulation:lateral_control:state_trajectory} shows the state and input trajectories of the various controllers across time,
further highlighting the similarity between the tuned controller and the the best trajectory.

\begin{figure}
\centering
\ifcompile
    \input{graphics/car_states/car_states.tex}
\else
    \includegraphics{graphics/car_states/car_states_main.pdf}
\fi
\caption{State and input trajectories of different controllers across time.}
\label{fig:simulation:lateral_control:state_trajectory}
\end{figure}

Finally, we use the technique described in \cref{section:imperfect_learning} to obtain a probabilistic bound on the norm $\|\J_\C(p^*)\|$, where $p^*$ denotes the optimal parameter.
\cref{fig:simulation:lateral_control:uncertainty} depicts the predicted upper-bound for a given $\epsilon$ (and $\beta$ fixed to $10^{-10}$) across a range of uncertainty radii, that is, for different values of $\|\theta-\theta^0\|$.
Selecting $\epsilon$ within the range $0.01$-$0.02$ yields an upper bound that closely matches the true norm (shown in red) without excessive conservatism.
Conversely, larger values of $\epsilon$ lead to an underestimation of the true norm (see inset), highlighting the tradeoff between tightness of the bound and statistical confidence.

\begin{figure}
\centering
\ifcompile
    \input{graphics/car_uncertainties/car_uncertainties.tex}
\else
    \includegraphics{graphics/car_uncertainties/car_uncertainties_main.pdf}
\fi
\caption{Upper bound on gradient norm for different confidence levels and uncertainty radii.}\label{fig:simulation:lateral_control:uncertainty}
\end{figure}

%% file: graphics/no_ce_tab/no_ce_tab.tex
\begin{tikzpicture}

\begin{axis}[
height=5.5cm,
width=0.95\columnwidth,
legend cell align={left},
legend style={draw=none, fill opacity=0.8, draw opacity=1, text opacity=1, font=\footnotesize, row sep =-2pt},
tick align=outside,
tick pos=left,
xlabel={Experiment},
xmin=-0.544, xmax=9.544,
xtick={0,1,2,3,4,5,6,7,8,9},
xticklabels={1,2,3,4,5,6,7,8,9,10},
ylabel={Cost},
ymin=0, ymax=84098.9911014497,
ylabel style = {font=\footnotesize},
tick label style={font=\footnotesize},
xlabel style = {font=\footnotesize},
legend image post style={scale=0.6},
]
\draw[draw=draw_color_table,fill=best_color_table,line width=0.24pt] (axis cs:-0.35,0) rectangle (axis cs:0.35,3928.6873704377);
\addlegendimage{area legend,draw=draw_color_table,fill=best_color_table,line width=0.24pt}
\addlegendentry{Best (training)}

\draw[draw=draw_color_table,fill=best_color_table,line width=0.24pt] (axis cs:0.65,0) rectangle (axis cs:1.35,3833.24636388634);
\draw[draw=draw_color_table,fill=best_color_table,line width=0.24pt] (axis cs:1.65,0) rectangle (axis cs:2.35,4514.75110707883);
\draw[draw=draw_color_table,fill=best_color_table,line width=0.24pt] (axis cs:2.65,0) rectangle (axis cs:3.35,4921.54471239633);
\draw[draw=draw_color_table,fill=best_color_table,line width=0.24pt] (axis cs:3.65,0) rectangle (axis cs:4.35,4680.44238930516);
\draw[draw=draw_color_table,fill=best_color_table,line width=0.24pt] (axis cs:4.65,0) rectangle (axis cs:5.35,4248.74173232943);
\draw[draw=draw_color_table,fill=best_color_table,line width=0.24pt] (axis cs:5.65,0) rectangle (axis cs:6.35,4402.70911493147);
\draw[draw=draw_color_table,fill=best_color_table,line width=0.24pt] (axis cs:6.65,0) rectangle (axis cs:7.35,4786.02083701967);
\draw[draw=draw_color_table,fill=best_color_table,line width=0.24pt] (axis cs:7.65,0) rectangle (axis cs:8.35,5065.8873775688);
\draw[draw=draw_color_table,fill=best_color_table,line width=0.24pt] (axis cs:8.65,0) rectangle (axis cs:9.35,5335.33391522862);
\draw[draw=draw_color_table,fill=trained_color_table,line width=0.24pt] (axis cs:-0.35,3928.6873704377) rectangle (axis cs:0.35,4936.52940775485);
\addlegendimage{area legend,draw=draw_color_table,fill=trained_color_table,line width=0.24pt}
\addlegendentry{Trained}

\draw[draw=draw_color_table,fill=trained_color_table,line width=0.24pt] (axis cs:0.65,3833.24636388634) rectangle (axis cs:1.35,5082.62973501243);
\draw[draw=draw_color_table,fill=trained_color_table,line width=0.24pt] (axis cs:1.65,4514.75110707883) rectangle (axis cs:2.35,5624.7840123337);
\draw[draw=draw_color_table,fill=trained_color_table,line width=0.24pt] (axis cs:2.65,4921.54471239633) rectangle (axis cs:3.35,6007.03589469812);
\draw[draw=draw_color_table,fill=trained_color_table,line width=0.24pt] (axis cs:3.65,4680.44238930516) rectangle (axis cs:4.35,6237.4515855182);
\draw[draw=draw_color_table,fill=trained_color_table,line width=0.24pt] (axis cs:4.65,4248.74173232943) rectangle (axis cs:5.35,5799.44990038795);
\draw[draw=draw_color_table,fill=trained_color_table,line width=0.24pt] (axis cs:5.65,4402.70911493147) rectangle (axis cs:6.35,5817.72475987227);
\draw[draw=draw_color_table,fill=trained_color_table,line width=0.24pt] (axis cs:6.65,4786.02083701967) rectangle (axis cs:7.35,6103.41297903807);
\draw[draw=draw_color_table,fill=trained_color_table,line width=0.24pt] (axis cs:7.65,5065.8873775688) rectangle (axis cs:8.35,7194.77411410598);
\draw[draw=draw_color_table,fill=trained_color_table,line width=0.24pt] (axis cs:8.65,5335.33391522862) rectangle (axis cs:9.35,6435.96437069286);
\draw[draw=table_dim_gray,fill=table_dark_gray,line width=0.32pt,postaction={pattern=north east lines, pattern color=table_dim_gray}] (axis cs:-0.35,4936.52940775485) rectangle (axis cs:0.35,22580.9761264747);
\addlegendimage{area legend,draw=table_dim_gray,fill=table_dark_gray,line width=0.32pt,postaction={pattern=north east lines, pattern color=table_dim_gray}}
\addlegendentry{Untrained}

\draw[draw=table_dim_gray,fill=table_dark_gray,line width=0.32pt,postaction={pattern=north east lines, pattern color=table_dim_gray}] (axis cs:0.65,5082.62973501243) rectangle (axis cs:1.35,44296.9057115681);
\draw[draw=table_dim_gray,fill=table_dark_gray,line width=0.32pt,postaction={pattern=north east lines, pattern color=table_dim_gray}] (axis cs:1.65,5624.7840123337) rectangle (axis cs:2.35,47409.101591218);
\draw[draw=table_dim_gray,fill=table_dark_gray,line width=0.32pt,postaction={pattern=north east lines, pattern color=table_dim_gray}] (axis cs:2.65,6007.03589469812) rectangle (axis cs:3.35,18706.4403905087);
\draw[draw=table_dim_gray,fill=table_dark_gray,line width=0.32pt,postaction={pattern=north east lines, pattern color=table_dim_gray}] (axis cs:3.65,6237.4515855182) rectangle (axis cs:4.35,80094.2772394759);
\draw[draw=table_dim_gray,fill=table_dark_gray,line width=0.32pt,postaction={pattern=north east lines, pattern color=table_dim_gray}] (axis cs:4.65,5799.44990038795) rectangle (axis cs:5.35,34732.5177069686);
\draw[draw=table_dim_gray,fill=table_dark_gray,line width=0.32pt,postaction={pattern=north east lines, pattern color=table_dim_gray}] (axis cs:5.65,5817.72475987227) rectangle (axis cs:6.35,35061.4543173363);
\draw[draw=table_dim_gray,fill=table_dark_gray,line width=0.32pt,postaction={pattern=north east lines, pattern color=table_dim_gray}] (axis cs:6.65,6103.41297903807) rectangle (axis cs:7.35,51152.184013696);
\draw[draw=table_dim_gray,fill=table_dark_gray,line width=0.32pt,postaction={pattern=north east lines, pattern color=table_dim_gray}] (axis cs:7.65,7194.77411410598) rectangle (axis cs:8.35,33383.9066943932);
\draw[draw=table_dim_gray,fill=table_dark_gray,line width=0.32pt,postaction={pattern=north east lines, pattern color=table_dim_gray}] (axis cs:8.65,6435.96437069286) rectangle (axis cs:9.35,50060.3276450856);
\draw (axis cs:0,1964.34368521885) node[
  scale=0.55,
  text=draw_color_table,
  rotate=0.0
]{-0.26};
\draw (axis cs:0,13758.7527671148) node[
  scale=0.55,
  rotate=0.0,
  fill=white,
  fill opacity=0.7,
  inner sep=1.5pt,
  text=draw_color_table,
  text opacity=1,
]{+4.75};
\draw (axis cs:1,1916.62318194317) node[
  scale=0.55,
  text=draw_color_table,
  rotate=0.0
]{-0.33};
\draw (axis cs:1,24689.7677232903) node[
  scale=0.55,
  rotate=0.0,
  fill=white,
  fill opacity=0.7,
  inner sep=1.5pt,
  text=draw_color_table,
  text opacity=1,
]{+10.56};
\draw (axis cs:2,2257.37555353942) node[
  scale=0.55,
  text=draw_color_table,
  rotate=0.0
]{-0.25};
\draw (axis cs:2,26516.9428017758) node[
  scale=0.55,
  rotate=0.0,
  fill=white,
  fill opacity=0.7,
  inner sep=1.5pt,
  text=draw_color_table,
  text opacity=1,
]{+9.50};
\draw (axis cs:3,2460.77235619817) node[
  scale=0.55,
  text=draw_color_table,
  rotate=0.0
]{-0.22};
\draw (axis cs:3,12356.7381426034) node[
  scale=0.55,
  rotate=0.0,
  fill=white,
  fill opacity=0.7,
  inner sep=1.5pt,
  text=draw_color_table,
  text opacity=1,
]{+2.80};
\draw (axis cs:4,2340.22119465258) node[
  scale=0.55,
  text=draw_color_table,
  rotate=0.0
]{-0.33};
\draw (axis cs:4,43165.864412497) node[
  scale=0.55,
  rotate=0.0,
  fill=white,
  fill opacity=0.7,
  inner sep=1.5pt,
  text=draw_color_table,
  text opacity=1,
]{+16.11};
\draw (axis cs:5,2124.37086616472) node[
  scale=0.55,
  text=draw_color_table,
  rotate=0.0
]{-0.36};
\draw (axis cs:5,20265.9838036783) node[
  scale=0.55,
  rotate=0.0,
  fill=white,
  fill opacity=0.7,
  inner sep=1.5pt,
  text=draw_color_table,
  text opacity=1,
]{+7.17};
\draw (axis cs:6,2201.35455746573) node[
  scale=0.55,
  text=draw_color_table,
  rotate=0.0
]{-0.32};
\draw (axis cs:6,20439.5895386043) node[
  scale=0.55,
  rotate=0.0,
  fill=white,
  fill opacity=0.7,
  inner sep=1.5pt,
  text=draw_color_table,
  text opacity=1,
]{+6.96};
\draw (axis cs:7,2393.01041850983) node[
  scale=0.55,
  text=draw_color_table,
  rotate=0.0
]{-0.28};
\draw (axis cs:7,28627.798496367) node[
  scale=0.55,
  rotate=0.0,
  fill=white,
  fill opacity=0.7,
  inner sep=1.5pt,
  text=draw_color_table,
  text opacity=1,
]{+9.69};
\draw (axis cs:8,2532.9436887844) node[
  scale=0.55,
  text=draw_color_table,
  rotate=0.0
]{-0.42};
\draw (axis cs:8,20289.3404042496) node[
  scale=0.55,
  rotate=0.0,
  fill=white,
  fill opacity=0.7,
  inner sep=1.5pt,
  text=draw_color_table,
  text opacity=1,
]{+5.59};
\draw (axis cs:9,2667.66695761431) node[
  scale=0.55,
  text=draw_color_table,
  rotate=0.0
]{-0.21};
\draw (axis cs:9,28248.1460078892) node[
  scale=0.55,
  rotate=0.0,
  fill=white,
  fill opacity=0.7,
  inner sep=1.5pt,
  text=draw_color_table,
  text opacity=1,
]{+8.38};
\end{axis}

\end{tikzpicture}

%% file: graphics/ce_tab/ce_tab.tex
\begin{tikzpicture}

\begin{axis}[
legend cell align={left},
legend style={fill opacity=0.8, draw opacity=1, text opacity=1, draw=none, font=\footnotesize, row sep =-2pt},
height=5.5cm,
width=0.95\columnwidth,
tick align=outside,
tick pos=left,
xlabel={Experiment},
xmin=-0.544, xmax=9.544,
xtick={0,1,2,3,4,5,6,7,8,9},
xticklabels={1,2,3,4,5,6,7,8,9,10},
ylabel={Cost},
ymin=0, ymax=84098.9911014497,
ylabel style = {font=\footnotesize},
tick label style={font=\footnotesize},
xlabel style = {font=\footnotesize},
legend image post style={scale=0.6}
]
\draw[draw=black,fill=best_color_table,line width=0.24pt] (axis cs:-0.35,0) rectangle (axis cs:0.35,3928.6873704377);
\addlegendimage{area legend,draw=black,fill=best_color_table,line width=0.24pt}
\addlegendentry{Best (training)}

\draw[draw=black,fill=best_color_table,line width=0.24pt] (axis cs:0.65,0) rectangle (axis cs:1.35,3833.24636388634);
\draw[draw=black,fill=best_color_table,line width=0.24pt] (axis cs:1.65,0) rectangle (axis cs:2.35,4514.75110707883);
\draw[draw=black,fill=best_color_table,line width=0.24pt] (axis cs:2.65,0) rectangle (axis cs:3.35,4921.54471239633);
\draw[draw=black,fill=best_color_table,line width=0.24pt] (axis cs:3.65,0) rectangle (axis cs:4.35,4680.44238930516);
\draw[draw=black,fill=best_color_table,line width=0.24pt] (axis cs:4.65,0) rectangle (axis cs:5.35,4248.74173232943);
\draw[draw=black,fill=best_color_table,line width=0.24pt] (axis cs:5.65,0) rectangle (axis cs:6.35,4402.70911493147);
\draw[draw=black,fill=best_color_table,line width=0.24pt] (axis cs:6.65,0) rectangle (axis cs:7.35,4786.02083701967);
\draw[draw=black,fill=best_color_table,line width=0.24pt] (axis cs:7.65,0) rectangle (axis cs:8.35,5065.8873775688);
\draw[draw=black,fill=best_color_table,line width=0.24pt] (axis cs:8.65,0) rectangle (axis cs:9.35,5335.33391522862);
\draw[draw=black,fill=trained_color_table,line width=0.24pt] (axis cs:-0.35,3928.6873704377) rectangle (axis cs:0.35,4361.77931508042);
\addlegendimage{area legend,draw=black,fill=trained_color_table,line width=0.24pt}
\addlegendentry{Trained}

\draw[draw=black,fill=trained_color_table,line width=0.24pt] (axis cs:0.65,3833.24636388634) rectangle (axis cs:1.35,4191.45984499539);
\draw[draw=black,fill=trained_color_table,line width=0.24pt] (axis cs:1.65,4514.75110707883) rectangle (axis cs:2.35,5252.82902706692);
\draw[draw=black,fill=trained_color_table,line width=0.24pt] (axis cs:2.65,4921.54471239633) rectangle (axis cs:3.35,5543.17293747146);
\draw[draw=black,fill=trained_color_table,line width=0.24pt] (axis cs:3.65,4680.44238930516) rectangle (axis cs:4.35,5158.62927690174);
\draw[draw=black,fill=trained_color_table,line width=0.24pt] (axis cs:4.65,4248.74173232943) rectangle (axis cs:5.35,4633.05309347758);
\draw[draw=black,fill=trained_color_table,line width=0.24pt] (axis cs:5.65,4402.70911493147) rectangle (axis cs:6.35,4829.42590974158);
\draw[draw=black,fill=trained_color_table,line width=0.24pt] (axis cs:6.65,4786.02083701967) rectangle (axis cs:7.35,5190.63876666606);
\draw[draw=black,fill=trained_color_table,line width=0.24pt] (axis cs:7.65,5065.8873775688) rectangle (axis cs:8.35,5813.07551366022);
\draw[draw=black,fill=trained_color_table,line width=0.24pt] (axis cs:8.65,5335.33391522862) rectangle (axis cs:9.35,6021.71987671094);
\draw[draw=table_dim_gray,fill=table_dark_gray,line width=0.32pt,postaction={pattern=north east lines, pattern color=table_dim_gray}] (axis cs:-0.35,4361.77931508042) rectangle (axis cs:0.35,22580.9761264747);
\addlegendimage{area legend,draw=table_dim_gray,fill=table_dark_gray,line width=0.32pt,postaction={pattern=north east lines, pattern color=table_dim_gray}}
\addlegendentry{Untrained}

\draw[draw=table_dim_gray,fill=table_dark_gray,line width=0.32pt,postaction={pattern=north east lines, pattern color=table_dim_gray}] (axis cs:0.65,4191.45984499539) rectangle (axis cs:1.35,44296.9057115681);
\draw[draw=table_dim_gray,fill=table_dark_gray,line width=0.32pt,postaction={pattern=north east lines, pattern color=table_dim_gray}] (axis cs:1.65,5252.82902706692) rectangle (axis cs:2.35,47409.101591218);
\draw[draw=table_dim_gray,fill=table_dark_gray,line width=0.32pt,postaction={pattern=north east lines, pattern color=table_dim_gray}] (axis cs:2.65,5543.17293747146) rectangle (axis cs:3.35,18706.4403905087);
\draw[draw=table_dim_gray,fill=table_dark_gray,line width=0.32pt,postaction={pattern=north east lines, pattern color=table_dim_gray}] (axis cs:3.65,5158.62927690174) rectangle (axis cs:4.35,80094.2772394759);
\draw[draw=table_dim_gray,fill=table_dark_gray,line width=0.32pt,postaction={pattern=north east lines, pattern color=table_dim_gray}] (axis cs:4.65,4633.05309347758) rectangle (axis cs:5.35,34732.5177069686);
\draw[draw=table_dim_gray,fill=table_dark_gray,line width=0.32pt,postaction={pattern=north east lines, pattern color=table_dim_gray}] (axis cs:5.65,4829.42590974158) rectangle (axis cs:6.35,35061.4543173363);
\draw[draw=table_dim_gray,fill=table_dark_gray,line width=0.32pt,postaction={pattern=north east lines, pattern color=table_dim_gray}] (axis cs:6.65,5190.63876666606) rectangle (axis cs:7.35,51152.184013696);
\draw[draw=table_dim_gray,fill=table_dark_gray,line width=0.32pt,postaction={pattern=north east lines, pattern color=table_dim_gray}] (axis cs:7.65,5813.07551366022) rectangle (axis cs:8.35,33383.9066943932);
\draw[draw=table_dim_gray,fill=table_dark_gray,line width=0.32pt,postaction={pattern=north east lines, pattern color=table_dim_gray}] (axis cs:8.65,6021.71987671094) rectangle (axis cs:9.35,50060.3276450856);
\draw (axis cs:0,1964.34368521885) node[
  scale=0.55,
  text=black,
  rotate=0.0
]{-0.11};
\draw (axis cs:0,13471.3777207776) node[
  scale=0.55,
  rotate=0.0,
  fill=white,
  fill opacity=0.7,
  inner sep=1.5pt,
  text=black,
  text opacity=1,
]{+4.75};
\draw (axis cs:1,1916.62318194317) node[
  scale=0.55,
  text=black,
  rotate=0.0
]{-0.09};
\draw (axis cs:1,24244.1827782818) node[
  scale=0.55,
  rotate=0.0,
  fill=white,
  fill opacity=0.7,
  inner sep=1.5pt,
  text=black,
  text opacity=1,
]{+10.56};
\draw (axis cs:2,2257.37555353942) node[
  scale=0.55,
  text=black,
  rotate=0.0
]{-0.16};
\draw (axis cs:2,26330.9653091424) node[
  scale=0.55,
  rotate=0.0,
  fill=white,
  fill opacity=0.7,
  inner sep=1.5pt,
  text=black,
  text opacity=1,
]{+9.50};
\draw (axis cs:3,2460.77235619817) node[
  scale=0.55,
  text=black,
  rotate=0.0
]{-0.13};
\draw (axis cs:3,12124.8066639901) node[
  scale=0.55,
  rotate=0.0,
  fill=white,
  fill opacity=0.7,
  inner sep=1.5pt,
  text=black,
  text opacity=1,
]{+2.80};
\draw (axis cs:4,2340.22119465258) node[
  scale=0.55,
  text=black,
  rotate=0.0
]{-0.10};
\draw (axis cs:4,42626.4532581888) node[
  scale=0.55,
  rotate=0.0,
  fill=white,
  fill opacity=0.7,
  inner sep=1.5pt,
  text=black,
  text opacity=1,
]{+16.11};
\draw (axis cs:5,2124.37086616472) node[
  scale=0.55,
  text=black,
  rotate=0.0
]{-0.09};
\draw (axis cs:5,19682.7854002231) node[
  scale=0.55,
  rotate=0.0,
  fill=white,
  fill opacity=0.7,
  inner sep=1.5pt,
  text=black,
  text opacity=1,
]{+7.17};
\draw (axis cs:6,2201.35455746573) node[
  scale=0.55,
  text=black,
  rotate=0.0
]{-0.10};
\draw (axis cs:6,19945.4401135389) node[
  scale=0.55,
  rotate=0.0,
  fill=white,
  fill opacity=0.7,
  inner sep=1.5pt,
  text=black,
  text opacity=1,
]{+6.96};
\draw (axis cs:7,2393.01041850983) node[
  scale=0.55,
  text=black,
  rotate=0.0
]{-0.08};
\draw (axis cs:7,28171.411390181) node[
  scale=0.55,
  rotate=0.0,
  fill=white,
  fill opacity=0.7,
  inner sep=1.5pt,
  text=black,
  text opacity=1,
]{+9.69};
\draw (axis cs:8,2532.9436887844) node[
  scale=0.55,
  text=black,
  rotate=0.0
]{-0.15};
\draw (axis cs:8,19598.4911040267) node[
  scale=0.55,
  rotate=0.0,
  fill=white,
  fill opacity=0.7,
  inner sep=1.5pt,
  text=black,
  text opacity=1,
]{+5.59};
\draw (axis cs:9,2667.66695761431) node[
  scale=0.55,
  text=black,
  rotate=0.0
]{-0.13};
\draw (axis cs:9,28041.0237608983) node[
  scale=0.55,
  rotate=0.0,
  fill=white,
  fill opacity=0.7,
  inner sep=1.5pt,
  text=black,
  text opacity=1,
]{+8.38};
\end{axis}

\end{tikzpicture}

%% file: graphics/quad2/quad2.tex
\begin{tikzpicture}

\begin{groupplot}[
  anchor=north west,
  group style={
    group size=3 by 1,
    horizontal sep=0.9cm
  },
  width=0.32\textwidth,
  height=0.28\textwidth
]

\nextgroupplot[
  legend cell align={left},
  legend style={
    fill opacity=0.8,
    draw opacity=1,
    text opacity=1,
    at={(0.97,0.05)},
    anchor=south east,
    font=\footnotesize,
    row sep =-2pt,
    legend columns = 2,
  },
  title style={yshift=-1ex},
  tick align=outside,
  tick pos=left,
  title={\footnotesize \textbf{Roll}},
  xlabel={Time [s]},
  xmajorgrids,
  xmin=0, xmax=10.1,
  ylabel={Euler Angle [rad]},
  ymajorgrids,
  ymin=-0.358415644146169, ymax=0.337484094429626,
  ylabel style = {font=\footnotesize},
  tick label style={font=\footnotesize},
  xlabel style = {font=\footnotesize},
  height=4cm,
]
\addplot [untrained] table {quad_roll_untrained.dat};
\addlegendentry{Untrained}
\addplot [trained] table {quad_roll_trained.dat};
\addlegendentry{Trained}
\addplot [best] table {quad_roll_best.dat};
\addlegendentry{Best}

\nextgroupplot[
  scaled y ticks=manual:{}{\pgfmathparse{#1}},
  tick align=outside,
  tick pos=left,
  title={\footnotesize \textbf{Pitch}},
  xlabel={Time [s]},
  xmajorgrids,
  xmin=0, xmax=10.1,
  ymajorgrids,
  ymin=-0.358415644146169, ymax=0.337484094429626,
  ytick style={draw=none},
  yticklabels={},
  tick label style={font=\footnotesize},
  xlabel style = {font=\footnotesize},
  title style={yshift=-1ex},
  height=4cm,
]
\addplot [untrained] table {quad_pitch_untrained.dat};
\addplot [trained] table {quad_pitch_trained.dat};
\addplot [best] table {quad_pitch_best.dat};

\nextgroupplot[
  scaled y ticks=manual:{}{\pgfmathparse{#1}},
  tick align=outside,
  tick pos=left,
  title={\footnotesize \textbf{Yaw}},
  xlabel={Time [s]},
  xmajorgrids,
  xmin=0, xmax=10.1,
  ymajorgrids,
  ymin=-0.358415644146169, ymax=0.337484094429626,
  ytick style={draw=none},
  yticklabels={},
  xlabel style = {font=\footnotesize},
  tick label style={font=\footnotesize},
  title style={yshift=-1ex},
  height=4cm,
]
\addplot [untrained] table {quad_yaw_untrained.dat};
\addplot [trained] table {quad_yaw_trained.dat};
\addplot [best] table {quad_yaw_best.dat};

\end{groupplot}

\begin{axis}[
  at={(group c1r1.south west)},
  anchor = north west,
  yshift=-3.5cm,
  width=0.88\textwidth,
  height=4cm,
  legend cell align={left},
  legend style={
    fill opacity=0.8,
    draw opacity=1,
    text opacity=1,
    at={(0.03,0.03)},
    anchor=south west,
    font=\footnotesize
  },
  ylabel style = {font=\footnotesize},
  tick label style={font=\footnotesize},
  xlabel style = {font=\footnotesize},
  tick align=outside,
  tick pos=left,
  xlabel={Time [s]},
  xmajorgrids,
  xmin=0, xmax=10.1,
  ylabel={Position [m]},
  ymajorgrids,
  ymin=-6.39924045583641, ymax=2.39996383123031,
  yshift=2cm,
  ylabel style={at={(-0.0705,0.5)}}
  ]
  \addplot [best_x] table {quad_x_best.dat};
  \addplot [trained_x] table {quad_x_trained.dat};
  \addplot [untrained_x] table {quad_x_untrained.dat};
  \addlegendentry{$x$}
  
  \addplot [best_y] table {quad_y_best.dat};
  \addplot [trained_y] table {quad_y_trained.dat};
  \addplot [untrained_y] table {quad_y_untrained.dat};
  \addlegendentry{$y$}
  
  \addplot [best_z] table {quad_z_best.dat};
  \addplot [trained_z] table {quad_z_trained.dat};
  \addplot [untrained_z] table {quad_z_untrained.dat};
  \addlegendentry{$z$}
\end{axis}

\end{tikzpicture}

%% file: graphics/quad1/quad1.tex
\begin{tikzpicture}

\begin{groupplot}[group style={group size=1 by 2, vertical sep = 15pt}]
\nextgroupplot[
legend cell align={left},
legend style={
    fill opacity=0.8, 
    draw opacity=1, 
    text opacity=1, 
    font=\footnotesize,
    row sep =-2pt,
    legend columns = 2,
},
tick align=outside,
tick pos=left,
xmajorgrids,
xmin=-2, xmax=50,
ylabel={Cost},
ymajorgrids,
ymin=980.946569658599, ymax=1335.66406227098,
width=7.5cm, height=4cm,
ylabel style = {font=\footnotesize},
tick label style={font=\footnotesize},
xtick style={draw=none},
xticklabels={},
ylabel style={at={(-0.15,0.5)}},
scaled y ticks=base 10:-2,
]
\addplot [iteration_exact]
table {%
0 1068.98958103349
1 1149.89920213064
2 1057.73524358108
3 1075.56820092132
4 1055.16891693307
5 1050.89468903128
6 1033.24522989258
7 1028.1134937134
8 1021.78655137795
9 1020.20122016105
10 1043.73463086466
11 1013.45235826377
12 1013.31688946459
13 1010.81212719701
14 1019.08848047498
15 1008.62612617758
16 1007.00355779954
17 1016.6553494027
18 1006.04470422419
19 1004.05049938465
20 1003.30902137322
21 1002.63689001796
22 1001.97227007209
23 1001.40165772535
24 1001.07125818665
25 1000.91892301157
26 1000.6205921814
27 1000.39592657442
28 1000.38255228839
29 1000.46354395964
30 1000.29505956218
31 1000.20566700265
32 1000.15992771192
33 1000.11715298445
34 1000.07655576484
35 1000.03824816908
36 1000.00211024647
37 999.967821228951
38 999.934802725355
39 999.903456904528
40 999.873429652913
41 999.844365299274
42 999.814652200572
43 999.771473067094
44 999.735955800132
45 999.700100031301
46 999.667218597534
47 999.636117407833
48 999.606543220906
49 999.572415998755
};
\addlegendentry{Exact model}
\addplot [iteration_inexact]
table {%
0 1319.54053987951
1 1130.51799228013
2 1054.0766531439
3 1048.32134984234
4 1045.23485137151
5 1046.53037263742
6 1039.95682821459
7 1027.60547666521
8 1022.0716164541
9 1018.59999650687
10 1016.65675978558
11 1016.14755482136
12 1013.59038735683
13 1015.66244884574
14 1009.79985661786
15 1011.72637769609
16 1007.73863979884
17 1021.29707698504
18 1007.77805835482
19 1004.13091609878
20 1003.39463909798
21 1008.44960392914
22 1001.84370029582
23 1001.30952646542
24 1000.94076264739
25 1000.77388300323
26 1000.64617416748
27 1000.51396888882
28 1000.30180771316
29 1000.35112086488
30 1000.13256579337
31 1000.08882195464
32 1000.06442728642
33 1000.00765656428
34 999.974275264021
35 999.949783707476
36 999.89732678339
37 999.865979193704
38 999.83610330151
39 999.807707253864
40 999.780627477698
41 999.754596955961
42 999.728969591207
43 999.69657011452
44 999.666246584668
45 999.639630558373
46 999.611261754529
47 999.585574417723
48 999.559609895658
49 999.533674668685
};
\addlegendentry{Inexact model}
\addplot [iteration_best]
table {%
0 997.070092050071
50 997.070092050071
};
\addlegendentry{Best achievable}

\nextgroupplot[
ylabel style = {font=\footnotesize},
tick label style={font=\footnotesize},
xlabel style = {font=\footnotesize},
tick align=outside,
tick pos=left,
xlabel={Iteration},
xmajorgrids,
xmin=-2, xmax=50,
xtick pos=left,
width=7.5cm, height=4cm,
ylabel={Estimation error},
ymajorgrids,
ymin=-0.00152655831483129, ymax=0.0471960831498623,
ylabel style={at={(-0.15,0.5)}},
ymode=log,
]
\addplot [estimation_error]
table {%
0 0.0449814176287399
1 0.0246352280577152
2 0.0146189350652025
3 0.0111059816701039
4 0.00975223353514615
5 0.00904218195866514
6 0.00852271865323181
7 0.00777631090355535
8 0.00703675814184719
9 0.0060748064080701
10 0.00502232657599654
11 0.00466113292253617
12 0.00400306347367971
13 0.00377673474854914
14 0.00329592230226074
15 0.00310049964647305
16 0.00270194322842253
17 0.00259758527173917
18 0.00245270962878607
19 0.002278097302501
20 0.00211868906842661
21 0.00203543801032554
22 0.00191200106399977
23 0.00177205376892065
24 0.00165981539544737
25 0.00155572046510964
26 0.0014664230392447
27 0.00138793943908598
28 0.00132261492176496
29 0.00126528786081048
30 0.00121130340478765
31 0.00116250846710527
32 0.00111735826071116
33 0.00107583455487605
34 0.00103822578383781
35 0.00100358588622931
36 0.000970658499328277
37 0.000940118182955826
38 0.000911301892623581
39 0.000884417931440329
40 0.000858954272193892
41 0.000835092774656742
42 0.00081245609867242
43 0.000791883000841253
44 0.000772030955712534
45 0.000753456942533511
46 0.000735767320587209
47 0.000718930260131937
48 0.000703100933180951
49 0.000688107206291145
};
\end{groupplot}

\end{tikzpicture}

%% file: graphics/quad3/quad3.tex
\begin{tikzpicture}

\begin{groupplot}[
  anchor=north west,
  group style={
    group size=3 by 1,
    horizontal sep=0.9cm
  },
  width=0.32\textwidth,
  height=0.28\textwidth
]

\nextgroupplot[
  legend cell align={left},
  legend style={
    fill opacity=0.8,
    draw opacity=1,
    text opacity=1,
    at={(0.97,0.05)},
    anchor=south east,
    font=\footnotesize,
    row sep =-2pt,
    legend columns = 1,
  },
  title style={yshift=-1ex},
  tick align=outside,
  tick pos=left,
  title={\footnotesize \textbf{Roll}},
  xlabel={Time [s]},
  xmajorgrids,
  xmin=0, xmax=10.1,
  ylabel={Euler Angle [rad]},
  ymajorgrids,
  ymin=-0.358415644146169, ymax=0.337484094429626,
  ylabel style = {font=\footnotesize},
  tick label style={font=\footnotesize},
  xlabel style = {font=\footnotesize},
  height=4cm,
]
\addplot [nominal] table {roll_nominal.dat};
\addlegendentry{Nominal}
\addplot [energy_informed] table {roll_energy.dat};
\addlegendentry{Energy-informed}

\nextgroupplot[
  scaled y ticks=manual:{}{\pgfmathparse{#1}},
  tick align=outside,
  tick pos=left,
  title={\footnotesize \textbf{Pitch}},
  xlabel={Time [s]},
  xmajorgrids,
  xmin=0, xmax=10.1,
  ymajorgrids,
  ymin=-0.358415644146169, ymax=0.337484094429626,
  ytick style={draw=none},
  yticklabels={},
  tick label style={font=\footnotesize},
  xlabel style = {font=\footnotesize},
  title style={yshift=-1ex},
  height=4cm,
]
\addplot [nominal] table {pitch_nominal.dat};
\addplot [energy_informed] table {pitch_energy.dat};

\nextgroupplot[
  scaled y ticks=manual:{}{\pgfmathparse{#1}},
  tick align=outside,
  tick pos=left,
  title={\footnotesize \textbf{Yaw}},
  xlabel={Time [s]},
  xmajorgrids,
  xmin=0, xmax=10.1,
  ymajorgrids,
  ymin=-0.358415644146169, ymax=0.337484094429626,
  ytick style={draw=none},
  yticklabels={},
  xlabel style = {font=\footnotesize},
  tick label style={font=\footnotesize},
  title style={yshift=-1ex},
  height=4cm,
]
\addplot [nominal] table {yaw_nominal.dat};
\addplot [energy_informed] table {yaw_energy.dat};

\end{groupplot}

\end{tikzpicture}

%% file: graphics/car_uncertainties/car_uncertainties.tex
\begin{tikzpicture}[spy using outlines={rectangle, magnification=2, size=1.2cm, connect spies}]

\pgfplotsset{001/.style={solid,mark_square,orange1}}
\pgfplotsset{002/.style={solid,mark_square,orange2}}
\pgfplotsset{003/.style={solid,mark_square,orange3}}
\pgfplotsset{01/.style={solid,mark_square,orange4}}
\pgfplotsset{true/.style={dashdot,mark_square,blue1}}

\begin{axis}[
legend cell align={left},
legend style={
  fill opacity=0.8,
  draw opacity=1,
  text opacity=1,
  at={(0.03,0.97)},
  anchor=north west,
  style={font=\footnotesize,row sep =-2.5pt},
  legend columns=2,
},
tick align=outside,
tick pos=left,
xlabel={Uncertainty Radius},
xmin=-0.025, xmax=0.525,
ylabel={Norm},
ymin=-0.5, ymax=68.9182087019307,
xmin=-0.01, xmax=0.505,
width=8.5cm,height=4.5cm,
ylabel style = {font=\footnotesize},
tick label style={font=\footnotesize},
xlabel style = {font=\footnotesize},
xmajorgrids,
ymajorgrids,
]
\addplot [001]
table {%
0.5 65.6422647606685
0.45 62.1379300858767
0.4 54.0295211173905
0.35 46.2368258448582
0.3 40.1862260201172
0.25 34.2255788403234
0.2 25.8639013557131
0.15 18.9349290208325
0.1 12.5020697256619
0.05 5.1798599529265
0 0.123385935793294
};
\addlegendentry{$\epsilon=0.01$}
\addplot [002]
table {%
0.5 56.3075815567722
0.45 53.2390420575323
0.4 46.8483392193192
0.35 39.2252132894562
0.3 35.1963904812522
0.25 29.1898054039487
0.2 22.8612319309468
0.15 17.0480577611259
0.1 11.0717566588771
0.05 4.50199004140488
0 0.123385935793294
};
\addlegendentry{$\epsilon=0.02$}
\addplot [003]
table {%
0.5 51.4815276452185
0.45 47.7755066661715
0.4 42.3043974799774
0.35 35.460780087086
0.3 32.3869779575422
0.25 25.891128278666
0.2 20.5761994591639
0.15 15.4259565984373
0.1 9.93834145832315
0.05 4.23912450355625
0 0.123385935793294
};
\addlegendentry{$\epsilon=0.03$}
\addplot [01]
table {%
0.5 48.8808879482961
0.45 45.5769010728037
0.4 40.1854016485533
0.35 34.1425120031434
0.3 30.6897828095978
0.25 24.2979854807564
0.2 19.8907982121144
0.15 14.6513600546647
0.1 9.4229710152531
0.05 4.06536148108757
0 0.123385935793294
};
\addlegendentry{$\epsilon=0.1$}
\addplot [true]
table {%
0.5 45.1464331295661
0.45 44.9980596786563
0.4 42.8631526806532
0.35 37.8990197699714
0.3 32.2803571283313
0.25 26.7692919087748
0.2 21.1903237736233
0.15 15.6538619054363
0.1 10.2390623615099
0.05 4.21463149500968
0 0.123385935425536
};
\addlegendentry{True}
\coordinate (spypoint) at (axis cs:0.4,40);
\coordinate (spyviewer) at (axis cs:0.48,30);

\end{axis}

\spy [black, line width = 0.5pt] on (spypoint) in node[anchor=north east] at (spyviewer);

\end{tikzpicture}

%% file: sources/conclusion.tex
We consider the problem of hyperparameter tuning for model predictive control (MPC).
We assume that the true system dynamics are unknown and affected by noise, and introduce a system identification procedure that operates alongside the parameter updates.
We then analyze the convergence properties of the proposed algorithm in two settings: (i) when the MPC prediction model is treated as a design variable, and (ii) when it is fixed to the best available model (certainty equivalence).
We develop an efficient method to bound the residual norm of the objective gradient in cases where the model is not exactly learned asymptotically.
Finally, we demonstrate the effectiveness of the approach through three simulation examples.
Though we do not focus on safety aspects, 
once a nominal controller is obtained using the proposed algorithms, the learned dynamical model can be combined with the robust design methodology proposed in \cite{zuliani2024closed} to ensure safety and robustness.
Future work will pursue this direction by directly focusing on strengthening safety certificates, for example by guaranteeing \emph{anytime feasibility}, that is, ensuring that the state constraints are satisfied for all iterations.

%% file: sources/convergence_proof_new.tex
The core of the proof involves showing that Assumption A in \cite{davis2020stochastic} is verified,
and then leveraging Theorem 1 in \cite{davis2020stochastic}.
For completeness, we report the assumption here for an algorithm of the form
\begin{align*}
p^{k+1} = p^k + \alpha_k[d^k+\xi^k],~~ d^k\in G^k(p^k),
\end{align*}
where $d^k$ represents a subgradient that will be specified later,
and $\xi^k$ denotes an error term.
Let $G:\R^{n_p} \rightrightarrows \R^{n_p}$ denote the conservative Jacobian of the cost that is to be minimized.
\begin{assumption}[\hspace{1sp}{\cite[Assumption A]{davis2020stochastic}}]\label{ass:davis} \hspace{1sp}
\begin{enumerate}
    \item All limit points of $\{p^k\}$ lie in $\mathcal{P}$.
    \item The iterates are bounded, that is, $\sup_{k \geq 1}~\|p^k\|< \infty$ and $\sup_{k \geq 1}\|d^k\|< \infty$.
    \item $\sum_{k\in\N} \alpha_k = \infty$ and $\sum_{k\in\N} \alpha_k^2 < \infty$.
    \item $\sum_{k\in\N} \alpha_k \xi^k < \infty$.
    \item For any unbounded increasing sequence $\{ k_j \} \subset \N$ such that $p^{k_j}\to \bar{p}$, it holds
    \begin{align*}
    \lim_{n\to \infty} \operatorname*{dist} \left( \frac{1}{n} \sum_{j=1}^{n} d^{k_j},G(\bar{p}) \right) = 0.
    \end{align*}
\end{enumerate}
\end{assumption}
We start by proving that $\J_{\bar{\mathcal{C}}}^k$ represents a ``sample'' of the true Jacobian $\J_\C$ of $\C$.
\begin{lemma}\label{lemma:exchange_diff_and_exp}
Under \cref{ass:local_lip_and_definability}, the expected cost $\C(p):=\E_v[\Cb(p,v)]$ is locally Lipschitz and definable with conservative Jacobian $\J_\C(p)=\E_v[\J_\Cb(p,v)]$.
\end{lemma}
\begin{proof}
The map $p,v\mapsto\Cb(p,v)$ is locally Lipschitz and definable because it is the composition of functions sharing these properties. 
Definability and local Lipschitz continuity are preserved by integration \cite{speissegger1999pfaffian}. Next, the function $\Cb(p,\cdot)$ is integrable for any $p$ because it is locally Lipschitz continuous. Moreover, $\J_\Cb(\cdot,v)$ is a conservative Jacobian for $\Cb(\cdot,v)$ by construction. For any given $v$, the elements of the conservative Jacobian $\J_\Cb(p,v)$ are all bounded for $p$ in compact sets, since $\J_\Cb(\cdot,v)$ is piecewise smooth. We can then invoke \cite[Theorem 3.10]{bolte2023subgradient}, which proves that expectation and path-differentiation can be exchanged, proving the claim.
\end{proof}
We define $J_{x_t}^k(\theta)$, $J_{u_t}^k(\theta)$, and $J_{y}^k(\theta)$ as elements of the true conservative Jacobians obtained via \cref{eq:backprop} using the true parameter $\theta$. Moreover, we let $J_{f,t}^k(\theta)$ and $J_{f,t}^k$ denote elements of $\J_{f}(x_t^k,u_t^k,\theta)$ and $\J_{f}(x_t^k,u_t^k,\theta^k)$, respectively. 
\begin{lemma}\label{lemma:bound_gradient}
Under \cref{ass:sys_id,ass:pe,ass:local_lip_and_definability,ass:P_bounded}, there exist some $L_1>0$ such that for all $k\in\N$ and any $p\in \mathcal{P}$ and $w$, we have
\begin{align*}
\operatorname{dist}(J^k_\Cb,\J_\Cb(p^k,v^k,\theta)) \leq L_1 \operatorname*{diam}(\Theta^k),
\end{align*}
with confidence at least $1-\delta$, where $\Theta^k$ is given in \cref{eq:c_k}.
\end{lemma}
\begin{proof}
Since $\operatorname*{dist}(J^k_\Cb,\J_\Cb^k(\theta)) \leq \operatorname*{dist}(J^k_\Cb,J_\Cb^k(\theta)) $ for any $J_\Cb^k(\theta)\in\J_\Cb^k(\theta)$, we focus on showing that $\|J_\Cb^k-J_\Cb^k(\theta)\| \leq L_1 \operatorname*{diam}(\Theta^k)$ for a specific $J_\Cb^k(\theta)$ that we define later and for some $L_1>0$. 
Given any $k\in\N$ and $t\in\Z_{[0,T-1]}$, let $e_{u_t}^k = J_{u_t}^k-J_{u_t}^k(\theta)$, $e_{x_t}^k = J_{x_t}^k-J_{x_t}^k(\theta)$, and $e_{y_t}^k = J_{y_t}^k-J_{y_t}^k(\theta)$. We have
\begin{align}
\|e_{x_{t+1}}^k\| & = \|J_{x_{t+1}}^k - J_{x_{t+1}}^k(\theta)\| \notag \\
& \leq \|J_{f,x,t}^k \| \| e_{x_t}^k \| + \|[J_{f,x,t}^k-J_{f,x}^k(\theta)] J_{x_t}^k(\theta)\| \notag \\
& \quad\,\, \|J_{f,u,t}^k\| \| e_{u_t}^k \| + \|[J_{f,u,t}^k-J_{f,u}^k(\theta)] J_{u_t}^k(\theta)\|. \label{eq:bound0}
\end{align}
Thanks to the local Lipschitz continuity of $x(\cdot)$ and $u(\cdot)$ (which follows from \cref{ass:local_lip_and_definability}), the boundedness of $\mathcal{P}$ by \cref{ass:P_bounded}, and the fact that $w_t$ is almost surely bounded by \cref{ass:sys_id}, we can find constants $M_x$, and $M_u$ satisfying $\|J_{x_t}^k(\theta)\| \leq M_x$ and $\|J_{u_t}^k(\theta)\| \leq M_u$ for all $t\in\Z_{[0,T-1]}$, and $k\in\N$. For the same reason, there exist constants $M_{f,x}$ and $M_{f,u}$ satisfying $\|J_{f,x,t}^k\| \leq M_{f,x}$ and $\|J_{f,u,t}^k\| \leq M_{f,u}$ for all $t\in\Z_{[0,T-1]}$, and $k\in\N$. Let $\bar{\mathcal{X}}$ be a set such that $x_t^k\in \bar{\mathcal{X}}$ for all $t\in\Z_{[0,T]}$ and $k\in\N$. From \cref{eq:bound0}, we have
\begin{multline}
\|e_{x_{t+1}}^k\| \leq M_{f,x} \|e_{x_t}^k\| + M_{f,u} \|e_{u_t}^k\| \\ + L_f (M_{u}+M_x)  \|\theta^k-\theta\|, \label{eq:bound1}
\end{multline}
where $f$ is Lipschitz (by \cref{ass:local_lip_and_definability}) in $\theta$ uniformly for $x\in \bar{\mathcal{X}}$ and $u\in \mathcal{U}$ with constant $L_f$. Similarly, we have
\begin{align}
\|e_{u_t}^k\| \leq L_\pi [ \|e_{x_t}^k \| + \|e_{y_t}^k\| ], \label{eq:bound2}
\end{align}
where $\pi$ is Lipschitz (by \cref{ass:local_lip_and_definability}) in $(x,y)$ uniformly for $p\in \mathcal{P}$ with constant $L_\pi$, and
\begin{align}
\|e_{y_t}^k\| \leq L_\text{MPC} [ \| e_{x_t}^k \| + \| e_{y_{t-1}}^k \| ], \label{eq:bound3}
\end{align}
where $\operatorname*{MPC}$ is Lipschitz (by \cref{ass:local_lip_and_definability}) in $(x,y)$ uniformly in $p$ for all $p\in \mathcal{P}$ with constant $L_\text{MPC}$. 
Since $x_0^k$ and $y_{-1}^k$ are known for all $k$, we have $e_{x_0}^k=e_{y_0}^k=0$. Suppose inductively that $e_{x_t}^k \leq L_{x,t} \|\theta^k-\theta\|$, and $e_{y_{t-1}}^k \leq L_{y,t-1} \|\theta^k-\theta\|$. Then from \cref{eq:bound2,eq:bound3} we have
\begin{align*}
\|e_{y_t}^k\| & \leq L_\text{MPC} ( L_{x,t} + L_{y,t-1} ) \| \theta^k-\theta \|,\\
\|e_{u_t}^k\| & \leq [ L_\pi L_{x,t} + L_\pi L_\text{MPC} (L_{x,t} + L_{y,t-1}) ] \| \theta^k-\theta \|.
\end{align*}
Letting $L_{y,t}=L_\text{MPC} ( L_{x,t} + L_{y,t-1} )$ and $L_{u,t}=L_\pi L_{x,t} + L_\pi L_\text{MPC} (L_{x,t} + L_{y,t-1})$, we obtain from \cref{eq:bound1} that
\begin{align*}
\|e_{x_{t-1}}^k\| \leq L_{x,{t+1}} \|\theta^k-\theta\|,
\end{align*}
where $L_{x,t+1}=M_{f,x} L_{x,t} + M_{f,u} L_{u,t} + L_f(M_u+M_x)$. This concludes the induction step, implying that for all $t\in\Z_{[0,T-1]}$ we have $\|e_{u_t}^k\| \leq L_u \|\theta^k-\theta\|$ and for all $t\in\Z_{[0,T]}$ we have $\|e_{x_t}^k\| \leq L_x \|\theta^k-\theta\|$.

Let $e_{x}^k=(e_{x_0}^k,\dots,e_{x_T}^k)$ and $e_u^k=(e_{u_0}^k,\dots,e_{u_{T-1}}^k)$. From \cref{ass:local_lip_and_definability}, we have that
\begin{align*}
\|J_\Cb^k-J_\Cb^k(\theta)\| \leq \|J_{\mathcal{C}}^k + c_3J_\text{p}^k\| [ T L_x + (T-1) L_u ] \|\theta^k-\theta\|,
\end{align*}
where $J_\mathcal{C}^k=J_{\mathcal{C}}(x^k,u^k,p^k)$ and $J_\text{p}^k\in\J_{\sum_{t=0}^{T} \operatorname*{dist}(\cdot_t,\mathcal{X})}(x^k)$. 
Since $\mathcal{C}$ and $\sum_{t=0}^{T} \operatorname*{dist}(\cdot_t,\mathcal{X})(\cdot)$ are locally Lipschitz, and $x^k$, $u^k$, and $p^k$ are almost surely bounded by \cref{ass:P_bounded,ass:sys_id}, there exists a constant $ L_C>0 $ such that $\|J_{\mathcal{C}}^k + c_3J_\text{p}^k\| \leq L_C$. Defining $L_1=L_CTL_x+L_C(T-1)L_u$ we have $\|J_\Cb^k-J_\Cb^k(\theta)\| \leq L_1 \|\theta^k-\theta\| \leq L_1 \operatorname*{diam}\Theta^k$ where the last statement holds with probability $1-\delta$. This completes the proof.
\end{proof}

\begin{lemma}
The update in \cref{eq:update} can be written as
\begin{align}\label{eq:proof1}
p^{k+1}= p^k+\alpha_k [d^k + \xi^k],~~ d^k\in G^k(p^k),
\end{align}
where
\begin{align*}
G^k(p) & = - \J_{\C}(p) - \alpha_k^{-1}\E_v[p-\alpha_k J_\Cb(p,v,\theta) - \Pi_{\mathcal{Y}^k}[p - \\ & \quad \alpha_k J_\Cb(p,v,\theta)]],\\
\alpha_k \xi^k & = \Pi_{\mathcal{Y}^k}[p^k-\alpha_k J_\Cb^k] - \E_v[\Pi_{\mathcal{Y}^k}[p^k-\alpha_k J_\Cb(p^k,v,\theta)]],
\end{align*}
where $J_\Cb(p,v,\theta)\in\J_\Cb(p,v,\theta)$.
\end{lemma}

\begin{proof}
By \cref{lemma:exchange_diff_and_exp}, let $J_\C^k = \mathbb{E}_v[J_\Cb(p^k,v,\theta)] \in \J_\C(p^k)$. Then substituting
\begin{align*}
d^k & = - J_\C^k-\alpha_k^{-1}\E_v[p^k-\alpha_k J_\Cb(p^k,v,\theta)] \\ & \quad + \alpha_k^{-1}\E_v[\Pi_{\mathcal{Y}^k}[p^k - \alpha_k J_\Cb(p^k,v,\theta)]]
\end{align*}
in \cref{eq:proof1} yields
\begin{align*}
p^{k+1} & =  p^k - \alpha_k J_\C^k - p^k + \alpha_k \E_v[J_\Cb(p^k,v,\theta)] \\ & \hphantom{={}} + \E_v[\Pi_{\mathcal{Y}^k}[p^k-\alpha_k J_\Cb(p^k,v,\theta)]] \\ & \hphantom{={}} + \Pi_{\mathcal{Y}^k}[p^k-\alpha_k J_\Cb^k] - \E_v[\Pi_{\mathcal{Y}^k}[p^k-\alpha_k J_\Cb(p^k,v,\theta)]] \\
& = p^k - \alpha_k J_\C^k - p^k + \alpha_k J_\C^k + \Pi_{\mathcal{Y}^k}[p^k-\alpha_k J_\Cb^k] \\
& = \Pi_{\mathcal{Y}^k}[p^k-\alpha_k J_\Cb^k ].\qedhere
\end{align*}
\end{proof}

\noindent Next, let
\begin{align*}
\alpha_k\eta_k:=\Pi_{\mathcal{Y}^k}[p_k-\alpha_k J_\Cb^k] - \Pi_{\mathcal{Y}^k}[p_k-\alpha_k J_\Cb^k(\theta)].
\end{align*}

\medskip
\begin{lemma}\label{lemma:first_error_limit}
The limit $\displaystyle \lim_{n \to \infty} \sum_{k=1}^{n} \alpha_k \eta_k$ exists and it is finite.
\end{lemma}
\begin{proof}
We use \cite[Proposition 4.32]{bonnans2013perturbation} by showing that all assumptions are verified. Let $k\in\N$ be fixed. Let $\Phi=\mathcal{Y}^k$, $f_k(p)=\|p-p^k+\alpha_k J_\Cb^k(\theta)\|^2$ and $g_k(p)=\|p-p^k+\alpha_k J_\Cb^k\|^2$. Observe that
\begin{enumerate}
    \item Since $\nabla^2 f_k(p)=2I$, for any $p,p'$
    \begin{align*}
    f_k(p') \geq f_k(p) + \nabla f_k(p) ^{\top} (p'-p) + \|p'-p\|^2,
    \end{align*}
    If $p\in\argmin_{x\in\Phi} f_k(x)$, then $\nabla f_k(p) ^{\top} d \geq 0$ for all locally feasible directions $d \neq 0$, i.e., those $d \neq 0$ such that there exists some $\epsilon>0$ for which $x+td\in \Phi$ for all $t\in (0,\epsilon]$ \cite[Theorem 4.9]{bagirov2014introduction}. Therefore, if $p'$ is sufficiently close to $p$, then $\nabla f_k(p)^{\top}(p-p') \geq 0$ and therefore $f$ satisfies the quadratic growth property
    \begin{align*}
    f_k(p') \geq f_k(p) + \|p-p'\|^2.
    \end{align*}
    \item We have
    \begin{align*}
    & \| \nabla f_k(p) - \nabla g_k(p) \| \\
    =~ & 2\| p - p^k + \alpha_k J_\Cb^k(\theta) - p + p^k - \alpha_k J_\Cb^k \| \\
    \leq~ & 2 \alpha_k [ \| J_\Cb^k(\theta)-J_\Cb^k \| ] \\
    \leq~ & 2 \alpha_k L_1 \operatorname*{diam}\Theta^k =: \kappa,
    \end{align*}
    where the last step follows from \cref{lemma:bound_gradient}. This means that $f_k-g_k$ is $\kappa$-Lipschitz.
\end{enumerate}
Let $S_0$ be the set of minimizers of the problem $\min_{p\in \Phi} ~f_k(p)$, and let $S_1$ be the set of minimizers of the problem $\min_{p\in \Phi}~ g_k(p)$. Observe that $\alpha_k \| \eta_k \| \leq \operatorname*{dist} (S_0,S_1) $. From \cite[Proposition 4.32]{bonnans2013perturbation} we have that for all $k$ large enough, say for all $k \geq \bar{k}$, $\operatorname*{dist}(S_0,S_1) \leq \kappa = 2 \alpha_k L_1 \operatorname*{diam}\Theta^k$.
We conclude that for all $k \geq \bar{k}$, $\alpha_k \|\eta_k\| \leq 2 \alpha_k L_1 \operatorname*{diam}\Theta^k$. By \cref{ass:stepsizes}, $\sum_{k=\bar{k}}^{\infty} 2 \alpha_k^2 < \infty$. 
By equation \cref{eq:c_k} in \cref{thm:sys_id}, the local Lipschitz continuity of $x(\cdot)$ and $u(\cdot)$, 
and the boundeness of $\mathcal{P}$, 
there exists some $C>0$ for which $\operatorname*{diam}\Theta^k \leq C \sqrt{\log k / k} $. 
Therefore, $\sum_{k=\bar{k}}^{\infty} 2 \alpha_k L_1 \operatorname*{diam}\Theta^k \leq 2C L_1 \sum_{k=\bar{k}}^{\infty} \alpha_k \sqrt{\log k / k} < \infty$. 
We conclude that $\lim_{n \to \infty} \sum_{k=1}^{n} \alpha_k \eta_k < \infty$.
\end{proof}

\begin{lemma}\label{lemma:second_error_limit}
The limit $\displaystyle \lim_{n \to \infty} \sum_{k=1}^{n} \alpha_k \xi^k$ exists.
\end{lemma}
\begin{proof}
We have $\alpha_k \xi^k = \alpha_k \eta_k + \alpha_k \varphi_k$, where $\alpha_k \varphi_k := \Pi_{\mathcal{Y}^k}[p^k-\alpha_k J_\Cb^k(\theta)] - \E_v[\Pi_{\mathcal{Y}^k}[p^k-\alpha_k J_\Cb(p^k,v,\theta)]]$. The limit $\lim_{n \to \infty} \sum_{k=1}^{n} \alpha_k \varphi_k$ exists by \cref{lemma:exchange_diff_and_exp} and \cite[Lemma A.5]{davis2020stochastic}, and $\alpha_k \eta_k$ is summable by \cref{lemma:first_error_limit}.
\end{proof}
\begin{lemma}\label{lemma:point_5_davis}
Point 5 of \cref{ass:davis} holds with
\begin{align*}
G(p) = -\J_\C (p) - N_{\mathcal{Y}}(p).
\end{align*}
\end{lemma}
\begin{proof}
Consider a fixed $v$ and $p_k$, and let
\begin{align*}
p^{k+1}(v)=\Pi_{\mathcal{Y}^k}[p^k-\alpha_k J_\Cb(p^k,v,\theta)]. 
\end{align*}
From the first-order optimality conditions we have
\begin{align*}
p^k - \alpha_k J_\Cb(p^k,v,\theta) - p^{k+1}(v)=w_k^\mathcal{Y}(v)\in N_{\mathcal{Y}^k}(p^{k+1}),
\end{align*}
where $N_{\mathcal{Y}^k}$ is the limiting normal cone of $\mathcal{Y}^k$.
Since all sets $\mathcal{Y}^k$ are bounded, and so are $p^k$ and $J_\Cb(p^k,v,\theta)$ for all $v$, there must exist some $M< \infty$ such that
\begin{align*}
\|w_k^\mathcal{Y}(v)\| = \|p^k - \alpha_k J_\Cb(p^k,v,\theta) - p^{k+1}(v)\| = M.
\end{align*}
Using the convexity of the set $G(p)$, we have
\begin{align*}
& \operatorname*{dist}\left( \frac{1}{n} \sum_{k=1}^{n}  -w_k^J - \alpha_k^{-1} w_k^\mathcal{Y}(v), G(p) \right) \\
& \leq \frac{1}{n} \sum_{k=1}^{n} \operatorname*{dist}\left( -w_k^J \! - \! \alpha_k^{-1} w_k^\mathcal{Y}(v), G(p) \right)
\end{align*}
where $-w_k^J\in \J_\C(p^k)$. Due to the outercontinuity of conservative Jacobians and \cref{ass:outer_semi}, we have $\operatorname*{dist}(w_k^J,\J_\C(\bar{p}))\to 0$ and $\operatorname*{dist}(w_k^{\mathcal{Y}}(v),N_\mathcal{Y}(\bar{p}))\to 0$ for any $v$, meaning that almost surely $\operatorname*{dist}( -w_k^J - \alpha_k^{-1} w_k^\mathcal{Y}(v), G(p) ) \to 0$.
By \cite[Lemma A.1]{davis2020stochastic}, we have that $\operatorname*{dist}\left( \frac{1}{n} \sum_{k=1}^{n}  -w_k^J - \alpha_k^{-1} w_k^\mathcal{Y}(v), G(p) \right)$ is uniformly dominated by an integrable function of $w$ (the proof is identical as the one of Claim 4, Page 150, \cite{davis2020stochastic}), which can be used in combination with the dominated convergence theorem to prove the claim (precisely as it is done in the final step of the proof on Page 150, \cite{davis2020stochastic}.
\end{proof}
\begin{proof}[Proof of \cref{thm:main}]
Points 1-5 in \cref{ass:davis} follow from the boundedness of each $\mathcal{Y}^k$, from \cref{ass:stepsizes}, and \cref{lemma:point_5_davis,lemma:second_error_limit}.
Therefore, Assumption A in \cite{davis2020stochastic} holds.
Additionally, Assumption B also holds because of \cref{ass:local_lip_and_definability}.
Invoking \cite[Theorem 1]{davis2020stochastic} concludes the proof.
\end{proof}

%% file: sources/path_diff.tex
\subsubsection{Path-differentiability of quadratic programs} \label{section:path_diff_qp}

For completeness, we provide sufficient conditions for the path-differentiability of a quadratic program of the form
\begin{align}
\begin{split}
\operatorname*{minimize}_x & \quad \frac{1}{2} x^\top Q(p) x + q(p)^\top x \\
\text{subject to} & \quad F(p)x = f(p),\\
& \quad G(p)x \leq g(p).
\end{split} \label{eq:QP}
\end{align}
where $p$ is a parameter. We require the following constraint qualification.
\begin{definition}\label{def:LICQ}
Let $x$ be an optimizer of \cref{eq:QP}. We say that the linear independence constraint qualification (LICQ) is satisfied at $x$ if the rows of $F(p)$ and the rows of $G(p)$ associated to active constraints (i.e., those rows $G_i(p)$ for which $G_i(p)x=g_i(p)$) are all linearly independent. More generally, if the constraints are given by $h(x)=0$ and $g(x) \leq 0$, then LICQ holds at $x$ if the rows of $\nabla_x h(x)$ and the rows of $\nabla_x g(x)$ associated to active constraints are linearly independent.
\end{definition}
\begin{assumption}\label{ass:qp}
Problem \cref{eq:QP} satisfies the linear-independence constraint qualification (LICQ) and $Q(p) \succ 0$ for all $p$. Moreover, the functions $Q$, $q$, $F$, $f$, $G$, and $g$ are locally Lipschitz and definable.
\end{assumption}
The Lagrangian and the dual problem associated to \cref{eq:QP} are given, respectively, by
\begin{align*}
\mathcal{L}(x,\lambda,\mu,p) & = \frac{1}{2} x^\top Q(p) x + q(p)^\top x + \mu^\top(F(p)x-f(p)) \\ & \quad + \lambda^\top (G(p)x-g(p)),
\end{align*}
and
\begin{align}
\begin{split}
\operatorname*{minimize}_{z=(\lambda,\mu)} & \quad \frac{1}{2} z^\top H(p) z + h(p)^\top z,\\
\text{subjec to} & \quad \lambda \geq 0,
\end{split} \label{eq:QP_dual}
\end{align}
respectively, where $H(p)$ and $h(p)$ are given by
\begin{align*}
H(p) & = \begin{bmatrix}
G Q^{-1} G^\top & G Q^{-1} F^\top\\
F Q ^{-1} G^\top & F Q ^{-1} F^\top
\end{bmatrix},~
h(p) = \begin{bmatrix}
GQ^{-1}q+g\\
FQ^{-1}q+f
\end{bmatrix},
\end{align*}
where we omitted the dependency on $p$ for simplicity. We define the following map, which retrieves the primal optimizer $y(p)$ (i.e., the solution of \cref{eq:QP}) given a dual optimizer $z(p)$ (i.e., the solution of \cref{eq:QP_dual})
\begin{align*}
\mathcal{G}(z,p) = -Q(p)^{-1} ( [F(p)^\top~G(p)^\top] z + q(p) ).
\end{align*}
\begin{theorem}[{\hspace{1sp}\cite[Theorem 1]{zuliani2023bp}}]
Under \cref{ass:qp}, the solution map $y(p)$ of \cref{eq:QP} is unique, locally Lipschitz and definable for all $p$. Moreover, we have that
\begin{align*}
W-Q(p)^{-1}[G(p)^\top~F(p)^\top]Z \in \J_y(p),
\end{align*}
where $W \in \J_{\mathcal{G},p}(z,p)$, with $z$ solving \cref{eq:QP_dual}, and $Z=-U^{-1}V$ with
\begin{align*}
U \in J_{P_C}(I-\gamma H(p))-I,~~V \in -\gamma J_{P_C}(Az+B),
\end{align*}
where $J_{P_C} = \operatorname*{diag}(\operatorname*{sign}(\lambda),\mathds{1}_{n_\text{eq}})$, and $A \in \J_H(p)$, $B \in \J_h(p)$.
\end{theorem}

\subsubsection{Path-differentiability of nonlinear optimization problems} \label{section:path_diff_nlp}

In this section we provide sufficient conditions for the path-differentiability of nonlinear optimization problems, extending the results of \cref{section:path_diff_qp} and allowing the utilization of nonlinear MPC formulations in \cref{eq:prob_nominal}.

We consider a parameterized nonlinear programming problem (NLP) in standard form
\begin{align}
\begin{split}
\operatorname*{minimize}_x & \quad f(x,p)\\
\text{subject to} & \quad h(x,p) = 0,\\
& \quad g(x,p) \leq 0,
\end{split}\label{eq:NLP}
\end{align}
where $f$, $h$, and $g$ are all definable and Lipschitz continuously differentiable, $x\in\R^{n_x}$ is the optimization variable, and $p\in\R^{n_p}$ is a parameter.

Assuming the existence of a multiplier vector, the KKT conditions for \cref{eq:NLP} are given by
\begin{align}
\begin{split}
\nabla_x \mathcal{L}(x,\lambda,\mu,p) & = 0\\
g(x,p) &\leq 0,\\
h(x,p) & = 0,\\
\lambda_i g_{i}(x,p) & = 0,~~ \forall i\in\Z_{[1,n_\text{in}]},\\
\lambda_i &\geq 0,~~ \forall i\in\Z_{[1,n_\text{in}]},
\end{split}\label{eq:KKT}
\end{align}
where $\mathcal{L}$ is the Lagrangian of \cref{eq:NLP} defined as
\begin{align*}
\mathcal{L}(x,\lambda,\mu,p) = f(x,p) + \sum_{i=1}^{n_\text{in}} \lambda_i g_i(x,p) + \sum_{j=1}^{n_\text{eq}} \mu_j h_j(x,p).
\end{align*}
We denote with $\mathcal{I}(x) \subset \Z_{[1,n_\text{in}]}$ the set of inequality constraints satisfied with equality at $x$, i.e., those indices $i$ for which $g_i(x)=0$. Additionally, we define $\mathcal{I}^\text{sa}(x) \subset \mathcal{I}(x)$ the set of strongly active inequality constraints, that is, the subset of $\mathcal{I}(x)$ for which the associated Lagrange multiplier $\lambda_i$ is strictly positive. We additionally define $\phi=(x,\lambda,\mu)\in\R^{n_\phi}$.

To ensure the path-differentiability of the solution map of the NLP, we need to introduce the following sufficient condition for optimality.
\begin{definition}[Strong second order sufficient conditions]
A primal-dual pair $\phi=(x,\lambda,\mu)$ satisfies the \emph{strong second order sufficient conditions of optimality} (SSOSC) if it satisfies \cref{eq:KKT}, and additionally $y^{\top} \nabla_x^2 \mathcal{L}(\phi,p) y >0$ for all $y\in \R^{n_x}$ satisfying $\nabla_x h(x,p)y=0$, and $\nabla_x g_{i}(x,p)y=0$ for all $i\in \mathcal{I}^\text{sa}(x)$.
\end{definition}
To ensure path-differentiability, we require the following.
\begin{assumption}\label{ass:path_diff}
The functions $g$, $h$, and $f$ in \cref{eq:NLP} are definable and twice continuously differentiable with Lipschitz continuous gradients jointly in $x$ and $p$. 
\end{assumption}
\begin{proposition}\label{prop:path_diff}
Under \cref{ass:path_diff}, if $\bar{\phi}$ is a primal-dual pair satisfying the SSOSC and LICQ conditions for a given $\bar{p}$, there exist a neighborhood $N$ of $\bar{p}$ and a locally Lipschitz definable function $\phi:N\to\R^{n_\phi}$ such that for all $p\in N$, $\phi(p)$ solves \cref{eq:NLP} with parameter $p$; moreover, $V\in \mathcal{J}_\phi(\bar{p})$ where $V$ is the unique solution of
\begin{align}
A(\bar{\phi},\bar{p})V=b(\bar{\phi},\bar{p}), \label{eq:derivative}
\end{align}
with
\begin{align*}
A(\bar{\phi},\bar{p})&=\begin{bmatrix}
\nabla_x^2 \mathcal{L}_p(\bar{\phi}) & \nabla_x g_{\mathcal{I}}(\bar{x},\bar{p})^{\top} & \nabla_x h(\bar{x},\bar{p})^{\top} \\
\nabla_x g_{\mathcal{I}}(\bar{x},\bar{p}) & 0 & 0 \\
\nabla_x h(\bar{x},\bar{p}) & 0 & 0
\end{bmatrix}, \\ 
b(\bar{\phi},\bar{p})&= \operatorname*{col}(%
\nabla^2_{px} \bar{\mathcal{L}}(\bar{\phi},\bar{p}),%
\nabla_p g_{\mathcal{I}}(\bar{x},\bar{p}),%
\nabla_p h(\bar{x},\bar{p})). 
\end{align*}
\end{proposition}

\begin{proof}
Under the SSOSC and LICQ, we can invoke \cite[Theorem 2.3.3]{jittorntrum1978sequential} to prove the existence of a Lipschitz continuous function $\phi:N \to \R^{n_\phi}$ defined in a neighborhood $N$ of $\bar{p}$ for which $\phi(p)$ solves \cref{eq:NLP} for all $p\in N$. Since the solution map is locally unique, it must be the unique solution of the KKT conditions \cref{eq:KKT}. By \cref{ass:path_diff}, leveraging the fact that gradients of continuously differentiable definable functions are definable \cite[Proposition B.7(3)]{vandenDries1996geometric}, each equality and inequality in \cref{eq:KKT} is a first-order formula in the sense of \cite{coste1999introduction}, and the set $\{ (p,\phi(p)): \phi(p) \text{ satisfies } \ref{eq:KKT} \}$ is definable by \cite[Theorem 1.13]{coste1999introduction} for all $p\in N$. Since $\phi(\cdot)$ is definable and locally Lipschitz, it is path differentiable \cite[Proposition 2]{bolte2021conservative}.

Next, if strict complementarity holds at $\bar{p}$ (that is, if $\lambda_i>0$ for all $i\in \mathcal{I}(s)$), then the result follows from \cite[Corollary 2.3.1]{jittorntrum1978sequential} by removing the rows of $\nabla_x g(\cdot,\bar{p})$ associated to inactive constraints and diving both sides by $\lambda_i$.

Suppose strict complementarity does not hold. In this case the set of active constraints of \cref{eq:NLP} may change for local variations of $p$. Let $\mathcal{I}^\text{wa}(\bar{x}):=\mathcal{I}(\bar{x})\cap\{i:\bar{\lambda}_i=0\}$ denote the set of weakly active constraints at $\bar{x}$. By \cite[Lemma 2.2.2]{jittorntrum1978sequential} all strongly active constraints $\mathcal{I}(\bar{x}) \setminus \mathcal{I}^\text{wa}(x)$ remain active for all values of $p\in N$. Therefore, locally, only constraints that belong to $\mathcal{I}^\text{wa}(\bar{x})$ can change (they can become strongly active, become inactive, or remain weakly active).

Let $\{ N_i \}_{i=1}^m$ be a partition of the full measure subset of $N$ where $\phi$ is differentiable, where each $N_i$ is associated to a different set of strongly active constraints $\mathcal{R}_{N_i}$ with $\mathcal{I}(x) \setminus \mathcal{I}^\text{wa}(x) \subseteq \mathcal{R}_{N_i} \subseteq \mathcal{I}(x)$. The Clarke Jacobian of $\phi$ at $\bar{p}$ is
\begin{align*}
\mathcal{J}^c_\phi(\bar{p}) = \operatorname*{conv}\{ \bigcup_{i=1,\dots,m,} \lim_{\substack{p \to \bar{p}\\p\in N_i}} \nabla \phi(p)  \}.
\end{align*}
For each $i$, the limit $\lim_{p \to \bar{p},~p\in N_i} \nabla \phi(p)$ satisfies \cref{eq:derivative} with $\nabla_x g_{\mathcal{I}}(\bar{x},\bar{p})$ replaced with $\nabla_x g_{\mathcal{R}_{N_i}}(\bar{x},\bar{p})$. By \cite[Theorem 2.2.2]{jittorntrum1978sequential}, $\mathcal{I}^\text{sa}=\mathcal{R}_{N_i}$ for some $i$, and the derivative obtained by solving $A(\bar{\phi},\bar{p})V=b(\bar{\phi},\bar{p})$ belongs to $\mathcal{J}_\phi^c(\bar{p})$. The result follows from $\mathcal{J}^c_\phi(\bar{p}) \subseteq \mathcal{J}_\phi(\bar{p})$ by \cite[Corollary 1]{bolte2021conservative}.
\end{proof}

%% file: sources/biography.tex
\vskip -2\baselineskip plus -1fil

\begin{IEEEbiography}[{\vspace*{-0.5cm}\includegraphics[height=1in,clip]{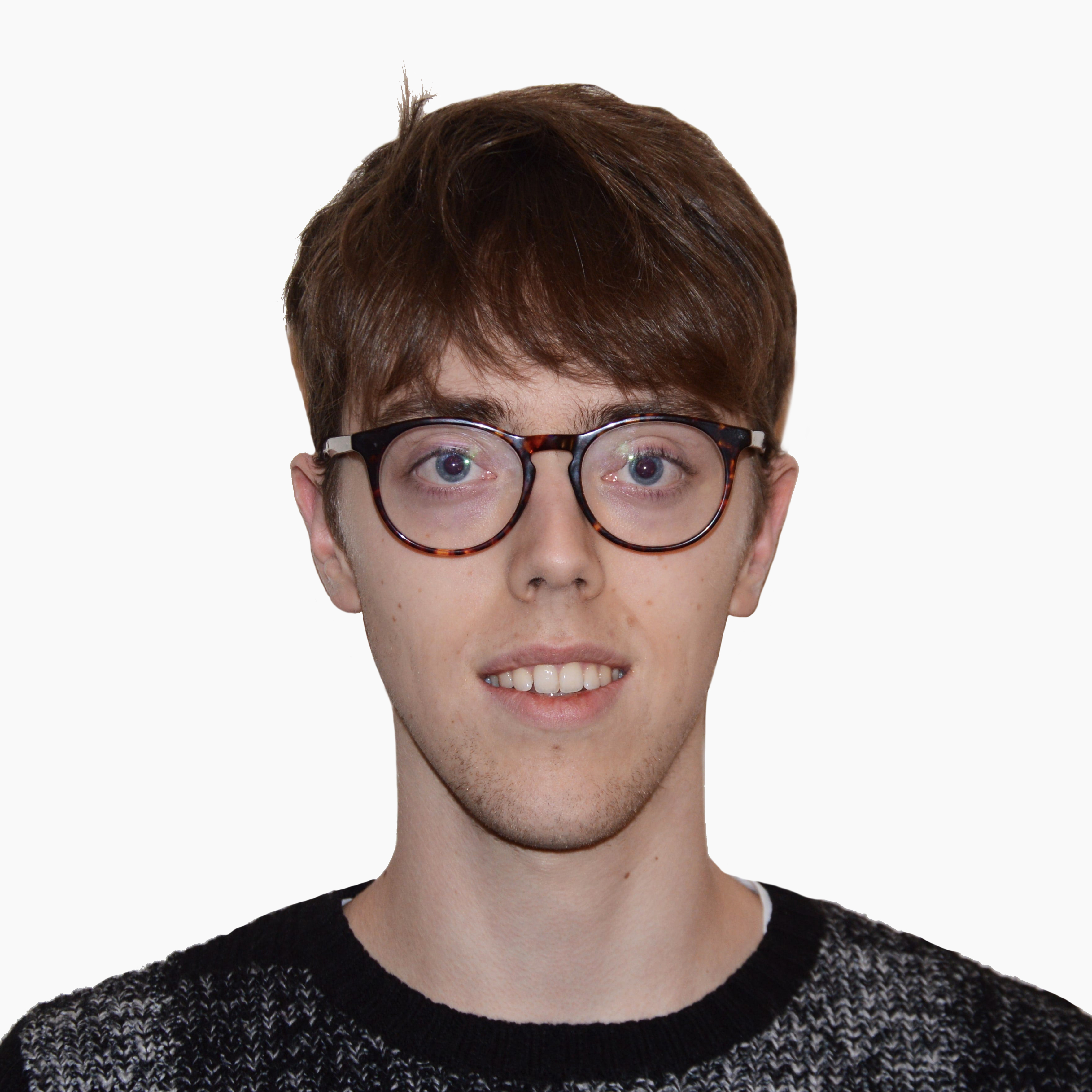}}]{Riccardo Zuliani} received the B.S. degree in mechatronics from the University of Padua (Italy), and the M.S. degree in robotics, systems, and control from ETH Zürich. He is currently pursuing his PhD at the Automatic Control Laboratory (IfA), ETH Zürich. His research interests include MPC, differentiable optimization, and additive manufacturing.
\end{IEEEbiography}

\vskip -3\baselineskip plus -1fil

\begin{IEEEbiography}[{\includegraphics[width=1in,height=1.25in,clip,keepaspectratio]{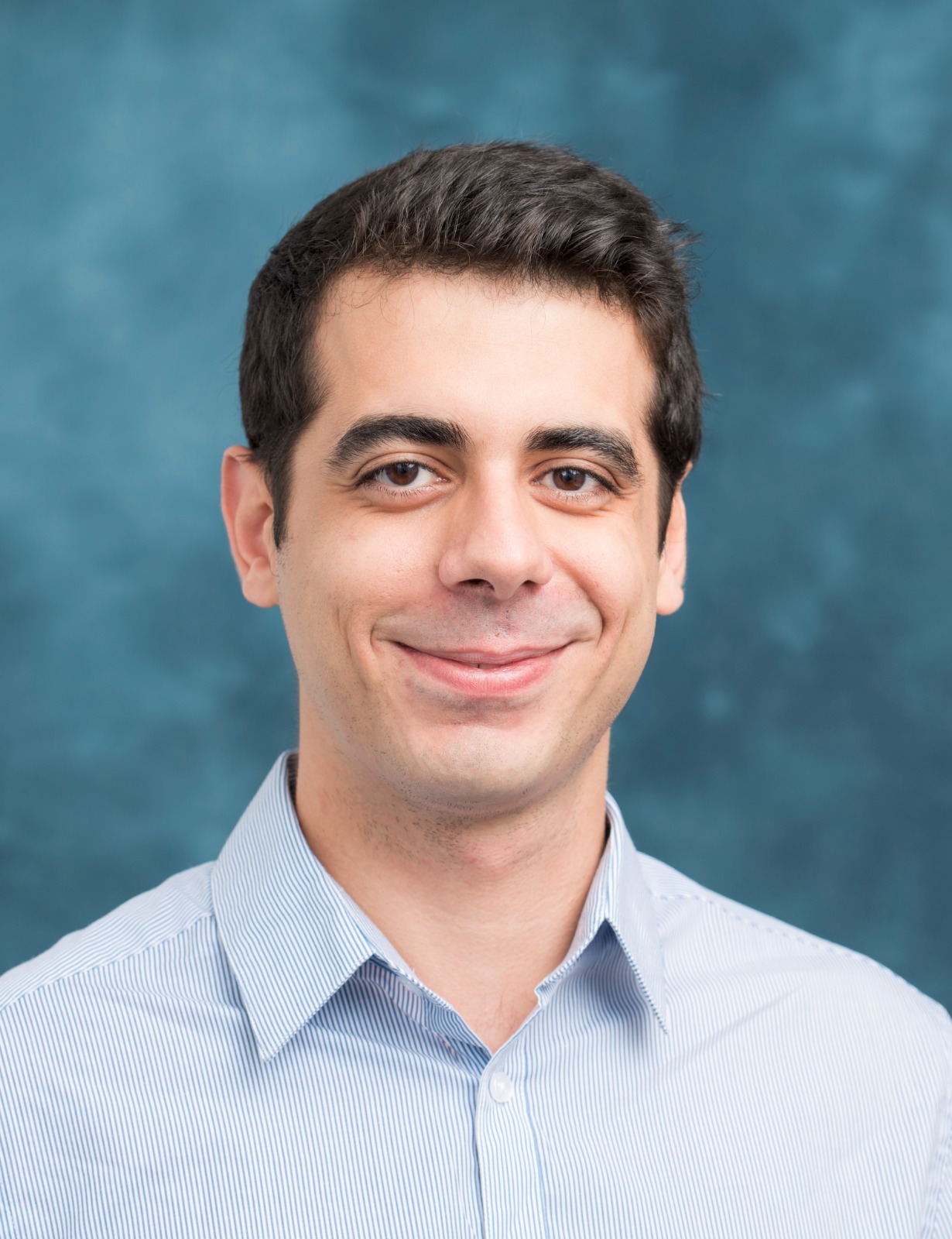}}]{Efe C. Balta} received the B.S.
degree in manufacturing engineering from the Faculty of Mechanical Engineering, Istanbul Technical University, in 2016, and the M.S. and Ph.D. degrees in mechanical engineering from the University of Michigan, Ann Arbor, MI, USA, in 2018 and 2021, respectively. He was a Post-Doctoral Researcher with the Automatic Control Laboratory (IfA), ETH Zürich between 2021 and 2023. Since September 2023, he has been leading the Control and Automation research group at inspire AG. His research interests include control theory, optimization, statistical learning, robotics, cyber-physical systems, and additive manufacturing.
\end{IEEEbiography}

\vskip -2\baselineskip plus -1fil

\begin{IEEEbiography}[{\includegraphics[width=1.05in,clip]{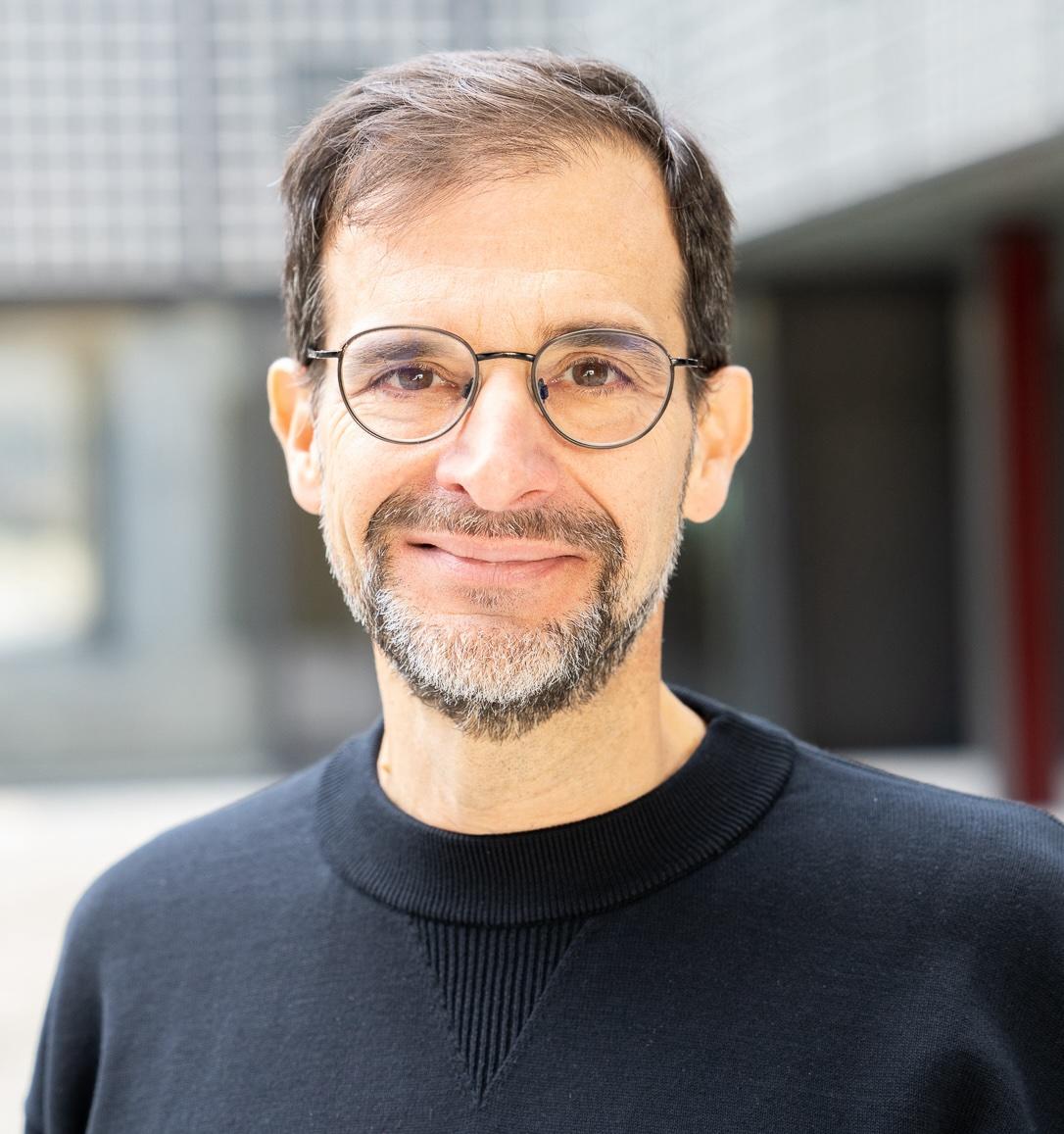}}]{John Lygeros} completed a B.Eng. degree in electrical engineering in 1990 and an M.Sc. degree in Systems Control in 1991, both at Imperial College of Science Technology and Medicine, London, U.K.. In 1996 he obtained a Ph.D. degree from the Electrical Engineering and Computer Science Department, University of California, Berkeley. During the period 1996-2000 he held research appointments at the National Automated Highway Systems Consortium, Berkeley, the Laboratory for Computer Science, M.I.T., and the Electrical Engineering and Computer Science Department at U.C. Berkeley. Between 2000 and 2003 he was a University Lecturer at the Department of Engineering, University of Cambridge, U.K., and a Fellow of Churchill College. Between 2003 and 2006 he was an Assistant Professor at the department of Electrical and Computer Engineering, University of Patras, Greece. In July 2006 he joined the Automatic Control Laboratory at ETH Zurich, where he is currently serving as the Head of the laboratory. His research interests include modelling, analysis, and control of hierarchical, hybrid, and stochastic systems, with applications to biochemical networks, transportation systems, energy systems, and industrial processes. John Lygeros is a Fellow of the IEEE, and a member of the IET and the Technical Chamber of Greece; between 2013 and 2023 he served as the Vice President for Finances and a Council Member of the international Federation of Automatic Control (IFAC), as well as on the Board of the IFAC Foundation.
\end{IEEEbiography}